\documentclass[a4paper,USenglish,cleveref,thm-restate]{lipics-v2021}

\nolinenumbers

\let\tilde\widetilde

\usepackage{onlyamsmath}
\usepackage{csquotes}
\usepackage{amsmath,amsthm}
\usepackage{thmtools}
\usepackage{xspace}
\usepackage{microtype}
\usepackage{booktabs}
\usepackage{fbox}
\usepackage{mathtools}
\usepackage{bbm}
\usepackage{xcolor}
\usepackage{nicefrac}
\usepackage{tabularx}

\let\originalleft\left
\let\originalright\right
\renewcommand{\left}{\mathopen{}\mathclose\bgroup\originalleft}
\renewcommand{\right}{\aftergroup\egroup\originalright}

\makeatletter
\let\c@author\relax
\makeatother
\usepackage[
doi = false,
giveninits = true,
isbn = false,
sortcites = true,
style = numeric,
url = false,
maxnames = 99,
minnames = 99,
]{biblatex}

\DeclareFieldFormat{titlecase}{#1}

\AtEveryBibitem{
\clearfield{month}
\clearlist{location}
\ifentrytype{book}{%
}{
\clearlist{publisher}
}
\ifentrytype{inproceedings}{%
\clearfield{series}
}{
}
\clearname{editor}
}

\def\nobreakbefore{\relax
\ifvmode\else
\ifhmode
\ifdim\lastskip > 0pt\relax
\unskip\nobreakspace
\fi
\fi
\fi
}
\let\oldcite\cite
\renewcommand\cite{\nobreakbefore\oldcite}

\makeatletter
\def\?#1{}
\def\whp{w.h.p\@ifnextchar-{.}{\@ifnextchar.{.\?}{\@ifnextchar,{.}{\@ifnextchar){.}{\@ifnextchar:{.:\?}{.\ }}}}}}

\makeatother

\usepackage{listings}
\makeatletter
\lst@AddToHook{OnEmptyLine}{\addtocounter{lstnumber}{-1}\vspace{-0.5\baselineskip}}

\lst@Key{output}{}{\def\lst@output{#1}}

\let\orig@lst@MakeCaption=\lst@MakeCaption
\def\lst@MakeCaption#1{%
\orig@lst@MakeCaption#1%
\ifx b#1%
\ifx\lst@output\@empty\else
\noindent
\expandafter\lst@makeoutputbox\expandafter{\lst@output}%
\vskip\belowcaptionskip
\fi
\fi
}

\def\lst@makeoutputbox#1{%

\vspace{-1pt}\makebox[\linewidth]{\fbox[LBR]{\parbox{\linewidth}{\small output function: #1}}}%
}

\let\oldlstinputlisting\lstinputlisting
\renewcommand{\lstinputlisting}[2][]{%
\par\noindent
\begin{minipage}{\linewidth}
\oldlstinputlisting[#1]{#2}%
\end{minipage}%
\par
}

\lst@AddToHook{Init}{%
\nolinenumbers
}

\lst@AddToHook{DeInit}{%
\modulolinenumbers
}

\makeatother

\crefname{listing}{Protocol}{Protocols}
\crefname{observation}{Observation}{Observations}

\newcommand{\ProtName}[1]{{\normalfont\textsc{#1}}}
\newcommand{\FieldName}[2]{\operatorname{\normalfont\texttt{#1}}_{#2}}

\newcommand{\GenerateRandomBit}{\ProtName{GenerateRandomBit}}
\newcommand{\Rand}[1]{\FieldName{rand}{#1}}
\newcommand{\Color}[1]{\FieldName{color}{#1}}
\def\black{\textsc{Black}}
\def\white{\textsc{White}}
\newcommand{\Balance}{\ProtName{Balance}}

\newcommand{\logred}[1]{\FieldName{logRed}{#1}}
\newcommand{\blue}[1]{\FieldName{blue}{#1}}

\newcommand{\EstLogN}{\ProtName{Approximate}}

\newcommand{\level}[1]{\FieldName{level}{#1}}
\newcommand{\group}[1]{\FieldName{group}{#1}}
\newcommand{\bclevel}[1]{\FieldName{bcLevel}{#1}}
\newcommand{\bcgroup}[1]{\FieldName{bcGroup}{#1}}

\newcommand{\lex}{\mathrm{lex}}
\newcommand{\tupleset}{\mathcal{S}}
\newcommand{\succlex}{\mathrm{succ}_\lex}
\newcommand{\FMV}{\mathrm{FMV}}

\newcommand{\broadcasttime}{\mathbf{B}}
\newcommand{\loadbalancingtime}{\mathbf{L}}
\newcommand{\hittingtime}{\mathbf{H}}

\newcommand{\isleader}[1]{\FieldName{isLeader}{#1}}
\newcommand{\terminated}[1]{\FieldName{terminated}{#1}}

\newcommand{\CountingProt}{\ProtName{CountExact}}

\newcommand{\Baseline}{\ProtName{UniformCounting}}
\newcommand{\val}[1]{\FieldName{val}{#1}}
\newcommand{\isactive}[1]{\FieldName{active}{#1}}

\newcommand{\Ev}[1]{\mathcal{E}_{\mathrm{#1}}}

\newcommand{\N}{\mathbb{N}}
\newcommand{\UniformDistr}{\mathcal{U}}
\newcommand{\Geom}{\mathrm{Geom}}

\DeclareMathOperator{\disc}{disc}
\DeclareMathOperator{\polylog}{polylog}

\def\floor#1{\left\lfloor #1\right\rfloor}
\def\ceil#1{\left\lceil #1\right\rceil}

\newcommand{\1}{\mathbbm{1}}

\newcommand{\pr}[1]{\ensuremath{\mathbf{Pr}\left[#1\right]}}
\newcommand{\E}[1]{\ensuremath{\mathbf{E}\left[#1\right]}}

\newcommand{\poly}[1]{\ensuremath{\mathrm{poly}\left(#1\right)}}

\newcommand{\mP}{\mathbf{P}}

\usepackage{xcolor}

\DeclarePairedDelimiter\abs{\lvert}{\rvert}
\DeclarePairedDelimiter\norm{\lVert}{\rVert}
\DeclarePairedDelimiterX{\inp}[2]{\langle}{\rangle}{#1, #2}

\makeatletter
\let\oldabs\abs
\def\abs{\@ifstar{\oldabs}{\oldabs*}}
\makeatother

\makeatletter
\let\oldnorm\norm
\def\norm{\@ifstar{\oldnorm}{\oldnorm*}}
\makeatother

\makeatletter
\let\oldinp\inp
\def\inp{\@ifstar{\oldinp}{\oldinp*}}
\makeatother

\newcommand{\LBT}{\ensuremath{\loadbalancingtime(G)}}
\newcommand{\BN}{\ensuremath{n \cdot 2^{b(\bclevel{})}}}
\newcommand{\PoofBound}{\ensuremath{2 \cdot 2^{b(\bclevel{})} + 3}}
\newcommand{\RedToBlue}{2^{b(\bclevel{})}}

\title{Counting in Population Protocols on Graphs} 

\author{Petra Berenbrink}{University of Hamburg, Germany}{petra.berenbrink@uni-hamburg.de}{https://orcid.org/0000-0002-6930-3259}{}
\author{Robert Elsässer}{University of Salzburg, Austria}{robert.elsaesser@plus.ac.at}{https://orcid.org/0000-0002-5766-8103}{}
\author{Tom Friedetzky}{Durham University, U.K.}{tom.friedetzy@durham.ac.uk}{https://orcid.org/0000-0002-1299-5514}{}
\author{Thorsten Götte}{University of Hamburg, Germany}{thorsten.goette@uni-hamburg.de}{https://orcid.org/0000-0001-9798-6993}{}
\author{Lukas Hintze}{University of Hamburg, Germany}{lukas.rasmus.hintze@uni-hamburg.de}{https://orcid.org/0009-0006-8348-4638}{}
\author{Dominik Kaaser}{Hamburg University of Technology, Germany}{dominik.kaaser@tuhh.de}{https://orcid.org/0000-0002-2083-7145}{}

\authorrunning{P.\ Berenbrink, R.\ Elsässer, T.\ Friedetzky, T.\ Götte, L.\ Hintze, and D.\ Kaaser}

\hideLIPIcs

\makeatletter
\def\authorrunning#1{%
\gdef\@authorrunning{\markright{\ifx\authoranonymous\relax\textcolor{black}{Anonymous author(s)} \else\if!#1!\textcolor{red}{Author: Please fill in the \string\authorrunning\space macro}\else#1\fi\fi}}}
\makeatother

\makeatletter
\newcommand{\DeclareMathActive}[2]{%
\expandafter\edef\csname keep@#1@code\endcsname{\mathchar\the\mathcode`#1 }
\begingroup\lccode`~=`#1\relax
\lowercase{\endgroup\def~}{#2}%
\AtBeginDocument{\mathcode`#1="8000}%
}

\newcommand{\std}[1]{\csname keep@#1@code\endcsname}
\patchcmd{\newmcodes@}{\mathcode`\-\relax}{\std@minuscode\relax}{}{\ddt}
\AtBeginDocument{\edef\std@minuscode{\the\mathcode`-}}

\DeclareMathActive{O}{\operatorname{\std{O}}}
\DeclareMathActive{o}{\@ifnextchar{(}{\operatorname{\std{o}}}{\std{o}}}
\makeatother

\def\paragraph#1{\subparagraph*{#1.}}

\newcommand{\Sample}{\ProtName{SampleLevelAndGroup}}

\def\ttrue{\text{\textsc{True}}}
\def\tfalse{\text{\textsc{False}}}
\newcommand{\issampling}[1]{\FieldName{isSampling}{#1}}
\newcommand{\groupbitssampled}[1]{\FieldName{groupBitsSampled}{#1}}

\newcommand{\rand}{\Rand}

\newcommand{\leveldone}[1]{\FieldName{levelDone}{#1}}

\newcommand{\TokenPos}[2]{W_{#1}^{(#2)}}
\newcommand{\vTokenPos}[1]{\vec{W}^{(#1)}}

\newcommand{\matMix}[2]{\mathbf{M}^{[#1,#2]}}
\newcommand{\mMix}[4]{\matMix{#3}{#4}_{#1,#2}}

\newcommand{\filter}[1]{\mathcal{F}^{(#1)}}

\newcommand{\iprob}[2]{P_{#1}^{(#2)}}
\newcommand{\jprob}[2]{Q_{#1}^{(#2)}}

\def\dcmGobble{}
\usepackage{silence}
\begin{document}

{\def\addcontentsline#1#2#3{}
\maketitle

\begin{abstract}
We consider the problem of counting the number of agents in a population protocol where the agents are connected by an underlying graph $G=(V,E)$ with $|V|=n$ nodes.
In each step, a random scheduler selects an edge uniformly at random, and the incident nodes make a state transition.
As per standard assumptions, agents are identical and anonymous, that is, have no identifiers.
To break symmetry, in each interaction one of the agents is declared as the initiator uniformly at random.
Our size counting protocol uses $\tilde O(n)$ states and stabilizes in $O( \broadcasttime(G) \cdot \log^2(n) + \loadbalancingtime(G) \cdot \log(n))$ interactions with high probability, where $\broadcasttime(G)$ is the broadcast time and $\loadbalancingtime(G)$ is the load balancing time.
Our protocol is based on novel protocols for sampling independent random bits (given that the scheduler determines an initiator and responder) and approximating $\log n$ up to an additive error of $O(\log \log n)$ with high probability.
The latter uses $O(\polylog(n))$ states and $O(\broadcasttime(G)\cdot\log^2 n)$ interactions.
Both results may be of independent interest.
The main protocol for exact counting requires the presence of a unique leader, the other two do not.
None of the protocols requires any knowledge about the graph $G$.
We conclude with impossibility results for terminating uniform population protocols that compute graph-size properties (like counting nodes or determining parity) with and without a leader.
\end{abstract}

\newpage
\enlargethispage{1cm}
{\tableofcontents}

}
\pagebreak

\section{Introduction}

The population protocol model \cite{DBLP:journals/dc/AngluinADFP06} is one of the most basic and fundamental models of distributed computing.
\Textcite{DBLP:journals/dc/AngluinADFP06} provide a comprehensive collection of examples that illustrate applications.
Population protocols are closely related to chemical reaction networks \cite{DBLP:journals/nc/SoloveichikCWB08}, and they model processes occurring naturally in life sciences \cite{liggett1985interacting,chen2013,cardelli2012}.
In classical population protocols, each agent is in one out of a fixed number of states. A probabilistic scheduler selects in each step a pair of agents independently and uniformly at random. One of the agents is called initiator and the other one the responder. These agents perform a state transition that depends only on their states.
In our paper, we assume that interactions are restricted to edges of an underlying graph $G = (V, E)$, where $V$ is the set of agents. We assume that in each step the scheduler randomly chooses one edge and determines randomly which of the endpoints acts as the initiator and which one as the responder.
Agents are anonymous and have neither identifiers nor knowledge of~$G$.

A large body of recent work (e.g., \cite{DBLP:conf/soda/AlistarhAEGR17, DBLP:conf/soda/AlistarhAG18,DBLP:conf/wdag/BerenbrinkEFKKR18, DBLP:conf/focs/DotyEGSUS21,DBLP:conf/stoc/BerenbrinkGK20,doi:10.1137/1.9781611978971.24,DBLP:journals/dc/AlistarhRV25})
proposes so-called \emph{non-uniform} protocols where the transition rules depend on the number of agents.
In biologically-inspired applications, however, this knowledge may not be available, and even polynomial approximations may not be feasible. In this context, size counting protocols \cite{DBLP:conf/podc/BerenbrinkKR19,DBLP:conf/podc/DotyE19,DBLP:journals/tcs/DotyE21,DBLP:conf/sand/DotyE22,DBLP:conf/podc/KaaserL24} can be used to count the number of agents before running a non-uniform protocol.
However, all the protocols in these papers assume that the underlying interaction graph is the complete graph.

Our main result is a novel population protocol to count the number of agents of the underlying (arbitrary) graph.
Its time complexity depends on bounds for standard broadcasting and load balancing on the underlying graph.
Its state complexity is asymptotically nearly optimal at $\tilde{\Theta}(n)$ with high probability (\emph{w.h.p.}, with probability at least $1 - n^{-\Omega(1)}$).
The protocol is \emph{graph-uniform}, i.e., neither its state space nor its transition function depend on $G$ (and thus it does not depend on $n$ either).
Informally, this means that the protocol works as advertised (from the initial configuration) without agents having to know anything about the graph.
The only hard requirement is the presence of a leader.

The high-level strategy is to first approximate $\log n$ to within a factor of $1 \pm o(1)$ w.h.p.\ using polylogarithmically many states. This part does not yet require the leader, and we believe this subroutine to be of independent interest.
Then we distribute $\tilde{\Theta}(n^2)$ tokens across the network and repeatedly balance that load until the system naturally converges
to a configuration where the loads of different agents differ by at most an additive constant.
From this, agents can directly output the exact value of $n$.
This technique was already used for exact counting in complete graphs (see \cite{DBLP:conf/wdag/DotyEMST18,DBLP:conf/podc/BerenbrinkKR19}) and crucially requires $\Omega(n^2)$ tokens to be balanced.

Our main technical contribution lies in our space-efficient implementation of the balancing: a na\"ive protocol would start with $\Omega(n^2)$ tokens located at the leader, and in each step two interacting agents would balance their loads.
This approach requires $\Omega(n^2)$ states to store the number of tokens.
To save states, we instead employ two types of tokens: \emph{red} and \emph{blue}.
The leader starts with $\Theta(n)$ \emph{red} tokens that are balanced (in a certain way).
Whenever an agent holds exactly one red token, it converts it into $\Theta(n)$ \emph{blue} tokens.
These blue tokens are then balanced using the na\"ive approach and are used to compute $n$.
As we ensure that each red token is (eventually) converted, the total number of blue tokens is (eventually) $\Omega(n^2)$ and thus large enough.
The crux of our protocol is the space-efficient balancing of red tokens:
Our protocol ensures that the number of red tokens on each agent is always a \emph{power of two}, so this number can be encoded using only $O(\log n)$ states instead of $n$ states.
This, however, entails that two agents holding red tokens cannot necessarily balance them.
The main difficulty in the analysis is overcoming the subtle dependencies between red tokens caused by our balancing protocol and showing that they indeed balance quickly.

\medskip

The stabilization times of our protocols depend heavily on properties of the underlying graphs.
The ``native'' representation of our results is in terms of the broadcast time $\broadcasttime(G)$ and the load balancing time $\loadbalancingtime(G)$ of the underlying graph. The former is the expected time it takes to send a message from an arbitrary node to every other node, assuming that in each step the message can be sent over a randomly chosen edge of $G$.
The latter is the time it takes to balance a polynomial number of arbitrarily distributed discrete tokens until every node has (almost) the same load.
In detail, our contributions are as follows.
\begin{enumerate}
\item An auxiliary protocol that estimates $\log(n)$ up to an additive error of $\log \log(n)$ in $O(\broadcasttime(G)\cdot \log^2 n)$ interactions using $O(\log^4(n))$ states w.h.p. This protocol does not need any knowledge of the graph nor does it require a leader.

\item Another auxiliary protocol for generating fair and independent random bits for population protocols on graphs.
This protocol also does not need any knowledge of the graph nor does it require a leader.
Instead, the protocol relies on the random initiator and responder choices by the scheduler.

\item As our main result, we present a protocol that counts the \emph{exact} population size in $O(\broadcasttime(G)\cdot \log^2 n + \loadbalancingtime(G) \cdot \log(n))$ interactions using $\tilde O(n)$ states w.h.p.
By \emph{counting} we mean that the population stabilizes to a configuration where each agent can derive the exact number of agents $n$ from its state.
Our protocol requires a leader, but does not use any information on the underlying graph.
In particular, neither $\broadcasttime(G)$ nor $\loadbalancingtime(G)$ are known or used by the protocol.

\item We prove that certain graph properties (e.g., size or parity) cannot reliably be computed by terminating population protocols, regardless of whether a leader is present or not.
In particular, this means that the fact that our exact counting protocol only stabilizes (and does not detect termination) is inevitable.
\end{enumerate}

\subsection{Model}

We consider a fixed undirected graph $G = (V,E)$ with $|V| = n$ agents and $m$ edges.
We define $\Delta(G)$, $\delta(G)$ and $d(G)$ as $G$'s maximum, minimum, and average node degree, respectively.
The state space of the agents is called $Q$ and each agent is initially in state $q_0\in Q$. An \emph{output function} $\rho \colon Q \rightarrow \Sigma^*$ maps states to an output domain.
The computation proceeds in discrete steps $t = 1,2, \ldots$\, In each step $t$, a stochastic scheduler picks an edge $(u, v)$ independently and uniformly at random with probability $1/(|E|)$ and then decides which of the nodes is the initiator and which one the responder. Note that this can be regarded as picking a directed edge. $u$ and $v$ observe each other's states and update their states according to a transition function $\Pi\colon Q\times Q \to Q\times Q$.

We call a population protocol \emph{(graph-)uniform} if neither its transition function~$\Pi$ nor the output function $\rho$ depend on (any parameters of) the underlying graph~$G$.
The size of the total state space $Q$ of our protocols may be viewed as unbounded (and must be for counting). When we say that a protocol uses $O(f(n))$ states w.h.p., we mean that there are state subsets $Q_1 \subseteq Q_2 \subseteq \ldots \subseteq Q$ with $|Q_n| = O(f(n))$ so that when the protocol is run on a graph of size $n$, all states used are in $Q_n$ w.h.p.
Alternatively, one could define all this in terms of the agents being multi-tape Turing machines, see \cite{DBLP:conf/podc/DotyE19}.

A configuration of the population protocol assigns each node a state.
We call a configuration $C$ \emph{stable} if no sequence of interactions exists that can change the future output $\rho(C(u))$ of any agent $u$.
The time complexity is measured by the so-called \emph{stabilization time}, which is the total number of interactions needed to reach a stable configuration.

In our results, we quantify the stabilization time in terms of the \emph{broadcast time} and the \emph{load balancing} time $\loadbalancingtime(G)$.
Let $T(v)$ be a random variable for the number of steps a broadcast (also called one-way epidemics \cite{DBLP:journals/dc/AngluinAE08a}) started at node $v$ needs to inform the entire graph.
We use the same definition as \textcite{DBLP:journals/dc/AlistarhRV25} and
define $\broadcasttime(G) = \max_{v} \{ \E{T(v)} \}$ to be the \emph{worst-case (over $v$) expected broadcast time} on $G$.

\Textcite{DBLP:journals/dc/AlistarhRV25} prove that $\broadcasttime(G)$ is the minimum of $O(mD + m\log n)$ and $O(m/\beta \cdot \log n)$ for any graph with diameter $D$ and edge expansion $\beta$.
Since $m \leq n^2$ and $D,\beta^{-1} \leq n$
we have $\broadcasttime(G) \in O(n^3)$ for any connected graph $G$. However, $\broadcasttime(G)$ can be as small as $O(n\log n)$, e.g., on a regular expander.

To define the load balancing time we first need some additional definitions. Let $\mP = \mP(G)$ be the matrix given by $\mP_{u,v} =
{1}/({2\Delta(G)})$ if $\{u, v\} \in E(G)$,
$\mP_{u,v} = 1 - {\deg(u)}/({2\Delta(G)})$ if $u = v$, and
$\mP_{u,v} = 0$ otherwise.

Let $\lambda_1(\mP) \geq \lambda_2(\mP) \geq \cdots \geq \lambda_n(\mP)$ be the $n$ eigenvalues of $\mP$, and define $\lambda(\mP) = \max\{|\lambda_2(\mP)|, |\lambda_n(\mP)|\}$.
Then, by \cite{sodaPaper}, it holds that
\[
\loadbalancingtime(G) = O\left(n \log(n) \cdot \frac{d(G)}{\Delta(G)}\cdot\frac{1}{1-\lambda(\mP(G))}\right).\]

As the spectral gap $(1-\lambda(\mP(G)))$ is bounded by $\Omega(n^{-3})$ we have $\loadbalancingtime(G) \in O(n^4 \log n)$ for any connected graph $G$. However, $\loadbalancingtime(G)$ can be as small as $O(n\log n)$, e.g., on a regular expander.
Finally, we require the notion of a graph's \emph{hitting time}.
For a graph $G=(V,E)$ and nodes $u,v\in V$, the hitting time $\hittingtime(u,v)$ from $u$ to $v$ is the expected time it takes a simple random walk on $G$, started at $u$, to visit $v$ for the first time. The hitting time $\hittingtime(G)$ of $G$ is then the maximal hitting time over all pairs $u,v \in V$.
It is known that the hitting time of any simple connected graph $G$ is at most $O(n^3)$ {\cite{levinmarkov}.}
We give some known upper bounds for $\broadcasttime(G)$, $\loadbalancingtime(G)$, and $\hittingtime(G)$ for some important graph classes in \cref{tab:blh_overview}.

\begin{table}
\caption{$\broadcasttime(G)$, $\loadbalancingtime(G)$, and $\hittingtime(G)$ for certain graph classes. Note that $\broadcasttime(G)$ and $\loadbalancingtime(G)$ are properties of population protocols on graphs while $\hittingtime(G)$ only depends on the graph. Note that $(1-\lambda(\mP(G)))$ denotes the spectral gap.}
\label{tab:blh_overview}
\centering
\begin{tabularx}{\textwidth}{llllX}
\toprule
\bfseries Graph Class & \bfseries $\broadcasttime(G)$ & \bfseries $\loadbalancingtime(G)$ & \bfseries $\hittingtime(G)$ & \\
\midrule
Clique & $O(n\log n)$& $O(n\log n)$ & $O(n)$ \\
Regular graphs & $O((\frac{n}{1-\lambda(\mP(G))})\log n)$ & $O((\frac{n}{1-\lambda(\mP(G))})\log n)$ & $O(\frac{n}{\sqrt{1-\lambda(\mP(G))}})$ \\
Cycle & $O(n^2\log n)$ & $O(n^3 \log n)$ & $O(n^2)$\\
$d$-dim. Torus ($d=2$) & $O(n^{3/2}\log n)$ & $O(n^2\log n)$ & $O(n \log n)$\\
$d$-dim. Torus ($d>2$) & $O(n \cdot d \cdot n^{1/d} \cdot \log n)$ & $O(n \cdot d \cdot n^{2/d} \cdot \log n)$ & $O(n)$\\
Worst Case & $O(n^{3})$ & $O(n^{4}\log n)$ & $O(n^{3})$ & \\
\bottomrule
\end{tabularx}
\end{table}

\subsection{Background and Related Work}

\paragraph{General Results for Population Protocols on Graphs}
\Textcite{DBLP:conf/opodis/AlistarhGR21} show how to simulate an arbitrary population protocol on a graph in the standard model (using a complete graph as the underlying graph). The simulation overhead is polylogarithmic
in the number of nodes, and quadratic in the conductance of the graph. The authors assume that a bound on the conductance is known in advance.
\Textcite{DBLP:conf/podc/ChenC20} present two self-stabilizing leader election protocols for $k$-regular graphs. The state complexity is polynomial in $k$ and they do not provide any bounds on the convergence time.
\Textcite{DBLP:conf/opodis/BeauquierBB13} present a constant state leader election protocol for general graphs. With the help of an oracle, this protocol can be made self-stabilizing. They focus on necessary conditions for self-stabilization and do not bound the convergence time.

\Textcite{DBLP:journals/dc/AlistarhRV25} consider leader election in population protocols on graphs.
They present several upper and lower bounds for the problem. For a graph $G$ they express their results in terms of the broadcast time $\broadcasttime(G)$ as defined before.
They present a space-efficient protocol using $O(\log n \cdot h(G))$ states where $h(G)\in O(\log n)$ is a parameter depending on the broadcast time. The expected convergence time is $O(\broadcasttime(G)\cdot \log n)$.
In that paper they also present a phase clock using $h+1$ states that triggers a clock tick every $\Theta(2^h)$ interactions.
We give an overview of leader election protocols in \cref{tab:le_overview}.

\Textcite{doi:10.1137/1.9781611978971.24} give upper and lower bounds on the majority problem in population protocols on graphs.
The upper bounds are expressed by the inverse of the spectral gap of the Markov chain
induced by the random scheduler on $G$
and the degree imbalance (defined as maximum divided by the minimum degree). For expander graphs their result matches the best result for complete graphs
by \textcite{DBLP:conf/focs/DotyEGSUS21}. Their result depends on an efficient phase clock implementation for graphs. For their phase clock they need graph parameters like the size of the graph or its relaxation time.

\smallskip

There are many results for the generation of random bits, approximate counting and counting in the standard population model where the nodes are connected by a complete graph. In the following, we present results for this classical model.

\paragraph{Approximate Counting}
One well-known way to approximate $\log n$ is to consider $n$ geometrically distributed random variables $X_1,\ldots, X_n$ with parameter $1/2$.
Their maximum $\max\{X_1,\ldots, X_n\}$ then approximates $\log n$ to within a constant factor.
This approach is easily realized by a population protocol \cite{DBLP:conf/soda/AlistarhAEGR17}.
\Textcite{DBLP:conf/podc/DotyE19} improve this result by repeating the experiment $O(\log n)$ times. Taking the average of $O(\log n)$ maxima then results in an approximation of $\log(n) \pm 5.7$.
Their algorithm requires $O(n \log^2n)$ interactions and $O(\log^4(n))$ states.
However, to obtain a polylogarithmic state space size, their algorithm relies on phase clocks to synchronize the agents. This requires the interaction frequencies of different agents to be the same (up to a multiplicative constant), and therefore is not directly transferable to population protocols on graphs.

\begin{table}
\caption{An overview of leader election algorithms for population protocols on graphs.}
\label{tab:le_overview}
\centering
\begin{tabularx}{\textwidth}{llllX}
\toprule
\bfseries Result & \bfseries Runtime & \bfseries Runtime(WC) & \bfseries States & \bfseries Comment \\
\midrule
\cite[Theorem 3]{DBLP:journals/dc/AlistarhRV25} & $O(\hittingtime(G) \cdot n\log n)$ & $O(n^4 \log n)$ & $O(1)$ & uniform, based on \cite{DBLP:conf/opodis/BeauquierBB13}. \\
\cite[Theorem 4]{DBLP:journals/dc/AlistarhRV25} & $O(\broadcasttime(G) + n\log n)$& $O(n^3)$ & $O(n^4)$ & uniform. \\
\cite[Theorem 5]{DBLP:journals/dc/AlistarhRV25} & $O(2^h\log n)$ & $O(n^3 \log n)$& $O(h\cdot\log n)$ & $h$ known, $2^h \geq \broadcasttime(G)$ \\
\bottomrule
\end{tabularx}
\end{table}

\Textcite{DBLP:conf/podc/BerenbrinkKR19} present an approximate counting protocol that computes $\lfloor \log n \rfloor$ or $\lceil \log n \rceil$.
Their protocol first elects a leader. The leader starts with one token.
In alternating phases, the tokens are doubled and balanced using load balancing. The first time there exists a node with two tokens after a balancing phase, the number of doubling steps is used as an approximation of $\log n$.
In a similar setting, \textcite{DBLP:journals/jacm/KowalskiM20} use load balancing for counting in the dynamic networks model.
Note that the dynamic network model is inherently synchronized, while the protocol from \cite{DBLP:conf/podc/BerenbrinkKR19} relies heavily on synchronizing the agents via the phase clock from \cite{DBLP:journals/dc/AngluinAE08a}.

\Textcite{DBLP:conf/sand/DotyE22,DBLP:conf/podc/KaaserL24} use geometrically distributed random variables to approximate $\log n$ in a dynamic variant of population protocols where agents arrive and depart during the execution.
\Textcite{DBLP:conf/sand/DotyE22} use the \emph{first missing value} of geometrically distributed random variables to approximate $\log n$.
Given $n$ geometrically distributed random variables $X_{1}, \dots, X_{n}$, the first missing value is defined as the smallest integer $i \notin \{ X_{1}, \dots, X_{n}\}$ (cf.\ \cite{louchard2005number,FLAJOLET1985182}).
W.h.p., the first missing value is in $\Theta(\log n)$ \cite{DBLP:conf/sand/DotyE22}.
While we follow a similar approach in our work, the bounds on the first missing value are not tight enough for our purposes (see the proof of Lemma 5.5 in the full version of \cite{DBLP:conf/sand/DotyE22}). Instead, we resort to a more fine-grained partitioning of the agents before we detect the first missing value. Our approach gives us an approximation in $[\log n - 2 \log \log n, \log n - \log \log n ]$ (see \cref{sec:approximate-counting}).

\paragraph{Exact Counting}
On general graphs, counting can be achieved with mass-carrying coalescing random walks: Each agent starts with one token. Whenever two agents holding tokens interact, one of them sends all its tokens to the other.
Eventually, there will be one agent holding all $n$ tokens left.
This trivial solution to exact counting on the complete graph requires $O(n)$ states and converges in $O(n^2)$ interactions \cite{DBLP:conf/podc/BerenbrinkKR19}.
On an arbitrary graph the running time is equal to the coalescence time (cf.\ \cite{DBLP:journals/talg/KanadeMS23}), which is at most $O(\hittingtime(G) \cdot n \cdot \log n )$.
For completeness, we give the corresponding algorithm and a sketch of the analysis in \cref{sec:folklore}.

\Textcite{DBLP:conf/wdag/DotyEMST18} count the exact number of agents without any prior size estimate. Their protocol converges in $O(n \log n \log\log n)$ interactions w.h.p.\ and in expectation, using $n^{60}$ states per agent.
\Textcite{DBLP:conf/podc/BerenbrinkKR19} present a protocol that computes $n$ in optimal $O(\log n)$ time using $\tilde O(n)$ states.
As their protocol for approximating $n$, the exact size counting protocol first elects a leader and then synchronizes phases via phase clocks. Hence, it cannot be applied in population protocols on graphs.
For further details on (exact) size counting protocols, we refer to the survey of size counting in population protocols by \textcite{DBLP:journals/tcs/DotyE21}.

A natural question is whether known \emph{fast} protocols for the complete graph, perhaps after some mild modifications (cf.\ \cite{DBLP:conf/opodis/AlistarhGR21}), can be applied directly on graphs.
Unfortunately, all currently known fast protocols for complete graphs rely on synchronization via so-called \emph{phase clocks} \cite{DBLP:journals/dc/AngluinAE08a}, and these phase clocks require network parameters such as the number of edges and the maximum degree.
To the best of our knowledge, all known phase clock implementations on graphs \cite{DBLP:journals/dc/AlistarhRV25, doi:10.1137/1.9781611978971.24, DBLP:conf/opodis/AlistarhGR21} require prior knowledge of graph parameters:
The clock from \cite{DBLP:journals/dc/AlistarhRV25} is driven by high-degree nodes of the graph, and its parameters must be adjusted based on the number of edges $m = |E|$ and the maximum degree $\Delta$.
The clock from \cite{doi:10.1137/1.9781611978971.24} requires prior knowledge of the spectral gap induced by the scheduler $(1-\lambda)$ and the minimum and maximum degrees.
The clock from \cite{DBLP:conf/opodis/AlistarhGR21} is defined for $d$-regular graphs only, and both the degree $d$ and the expansion $\beta$ must be known.
In contrast, our counting protocol can be applied without \emph{any} knowledge of graph parameters: without resorting to synchronization, our protocol converges to the correct output in a time that depends on spectral properties of the graph.

\paragraph{Random Bits}
The problem of generating random bits has received significant attention in both classical population protocols and population protocols on graphs.
On the complete graph, generating independent random bits is seemingly simple.
Each node is equipped with a so-called \emph{flip bit} which is flipped in each interaction.
The initiator then uses the flip bit of the responder as a random bit \cite{DBLP:conf/soda/AlistarhAEGR17, DBLP:conf/soda/BerenbrinkKKO18}.
This provides w.h.p.\ nearly random bits after some initialization time $n/2\pm O(\sqrt{\log(n)/n})$.
\Textcite{DBLP:conf/podc/SudoOIKM19} provide a simple protocol that ensures that the number of agents with flip bit $1$ equals exactly the number of agents with flip bit $0$.
Note that this approach applied on graphs requires that roughly half of each agent's neighbors have a \emph{flip bit} $0$ and the remaining ones hold $1$, which is not the case for low-degree graphs.

In population protocols on graphs the required randomness is, to our best knowledge, always derived from the initiator-responder pattern (see \cite{DBLP:journals/dc/AlistarhRV25}).
That is, in an interaction, the initiator samples bit $1$ and the responder samples bit $0$.
Since a node is the initiator or responder with equal probability, nodes always sample fair coins.
However, the coins produced by this approach are highly correlated as the coins of interacting nodes are always different.
Thus, the bits created by this protocol are fair but not independent.

\Textcite{kanaya_et_al:LIPIcs.OPODIS.2024.37} present a protocol to sample bits based on a static coloring of the graph. They begin by computing a coloring (for this they need an upper bound on $n$ or on the maximal degree). In each subsequent interaction \emph{only} the agent with the higher-ranked color obtains a random bit depending on its role (initiator/responder).
The other agent does not sample a bit.
This ensures independence of the sampled bits but means that there will be agents that cannot sample any bits as their color is a local minimum.

\section{Results}

In this section, we give an overview of our results, including descriptions of algorithms and sketches of the associated proofs.
We begin with our two main results for approximate and exact counting in \cref{sec:approxcountingshort} and \cref{sec:exactcountingshort}.
In \cref{sec:impossibleshort} we show that there can be no terminating protocol for counting that is correct with probability $1$.
Finally, in \cref{sec:random-bits-overview} we present an auxiliary protocol to sample independent, unbiased bits.
The full proofs are given in subsequent sections.

\subsection{Approximate Counting}\label{sec:approximate-counting}
\label{sec:approxcountingshort}
\Textcite{DBLP:conf/sand/DotyE22} use the \emph{first missing value} of $n$ geometrically distributed random variables with parameter $1/2$ to approximate $\log n$.
To improve their estimate of $\log n$ we consider $n$ tuples of random variables $(L_i, G_i)$ as follows. The random variables $L_i$ have a geometric distribution with parameter $1/2$.
The random variables $G_i$ have a uniform distribution over $1,2,\ldots,2^{\lceil \log L_i \rceil}$.
We say that the random variable $L_i$ determines a level, and the random variable $G_i$ determines a group within that level.
Now we order our tuples lexicographically and calculate the first missing tuple, which we call $(M_L,M_G)$ in the following.
We show that $M_L$ is a tight approximation of $\log n$; it only has an additive error of $\log\log n+O(1)$ w.h.p.\ (see \cref{lem:fmv_tuple_tailbound} for the detailed statement).

The idea of this approach is as follows. The random variable $L_i$ is geometrically distributed, hence the number of variables with a value of $k$ is roughly $n/2^k$. For all random variables $L_i$ with $L_i=L$ (belonging to level $L$) we draw a coupon from the range $1,2,\ldots 2^{\lceil \log L \rceil}$. We output the smallest level $L$ for which we could not collect all $2^{\lceil \log L \rceil}$ coupons of that level. The protocol for sampling $L_i$ and $G_i$ is straightforward and can be found in \cref{approx_alg}.
Note that we use the subprotocol \textsc{GenerateRandomBit} to generate the random bits for both the geometric and the uniform distribution. We explain the details of this protocol in \cref{sec:random-bits-overview}.

\medskip

Our protocol \EstLogN\ uses the idea presented above. In our protocol node $u$ computes $\level u$ with a geometric distribution (role of $L_i$) and $\group u$
with a uniform distribution (role of $G_i$).
Before we describe protocol $\EstLogN$ to compute the first missing tuple we need one more definition:
$\succlex(\bclevel{} , \bcgroup{} )$ returns the successor of the tuple $(\bclevel{} , \bcgroup{} )$ in lexicographical order.
The protocol is applied by agent $u$ as soon as $\level u$ and $\group u$ are calculated. In the protocol $\EstLogN$, node $u$ has two more fields $\bclevel u$ and $\bcgroup u$ (\emph{bc} stands for broadcast) which are initialized to $1$ and $0$, respectively.
Whenever node $u$ has $(\level u, \group u)=(\bclevel{u}, \bcgroup{u})$, node $u$ knows for sure that the tuple $(\bclevel{u}, \bcgroup{u})$ is \emph{not} missing and thus sets $(\bclevel{u}, \bcgroup{u})$ to the successor according to the lexicographical order.
This process continues as long as all tuples are present.
Crucially, for the first missing tuple, no further increase is possible, and the protocol stabilizes.

\begin{lstlisting}[caption={\EstLogN(u, v)}, label={alg:EstLogN}]
w.l.o.g. let $(\bclevel u, \bcgroup u) \geq_{\text{lex}} (\bclevel v, \bcgroup v)$ /*otherwise, exchange $u$ and $v$*/

if $\exists i \in \{u,v\} \colon (\level i, \group i) = (\bclevel u, \bcgroup u)$ then
    $(\bclevel u, \bcgroup u) \gets \succlex(\bclevel u, \bcgroup u)$ /*start looking for the next value*/

$(\bclevel v, \bcgroup v) \gets (\bclevel u, \bcgroup u)$ /*broadcast the maximum */
\end{lstlisting}

\begin{theorem}[restate=restateThmApproximateCounting,label={thm:approximate-counting}]
Protocol $\EstLogN$ uses $O(\log^4(n))$ states w.h.p.,
and stabilizes after $O(\log^2(n) \cdot \broadcasttime(G))$ interactions w.h.p.
When it stabilizes we have
$\bclevel{v}=\hat{\ell}$ for all agents $v$.
For $\hat{\ell}$ we have with probability $1-4/n^c$ that
\quad \[\displaystyle \hat{\ell} \in [\log n - 2 \log\log n - \log{c}, \quad \log n - \log\log n +4 + 2c ].\]
\end{theorem}

\paragraph{Proof Sketch for \cref{thm:approximate-counting}}

The subprotocol \textsc{GenerateRandomBit} which is used to sample from the geometric and uniform distribution uses a constant number of states.
Further, the four key
fields of \EstLogN, $\level{}$, $\group{}$, $\bclevel{}$, and
$\bcgroup{}$ are each $O(\log n)$ w.h.p., since the
maximum of $n$ independent $\mathrm{Geom}(1/2)$ random variables is $O(\log n)$
w.h.p., and group values are bounded by level values. The broadcast fields $\bclevel{}$ and $\bcgroup{}$
track at most one beyond the maximum sampled values, so they are $O(\log n)$ as well.
This gives us a state complexity of $O(\log^4 n)$ w.h.p.

\smallskip

The stabilization time is $O(\log^2 n \cdot \broadcasttime(G))$ w.h.p.
Each agent needs $O(\log n)$ random bits w.h.p.\
to determine its $\level{}$ and $\group{}$ (the geometric variable
for $\level{}$ dominates). \Cref{thm:bitsampling} shows that agents collect the required $O(\log n)$ random bits in time $O(\tfrac{|E(G)|}{\delta(G)}\log n)$. Note that $\tfrac{|E(G)|}{\delta(G)}$ is precisely the expected time that an agent with minimal degree $\delta(G)$ is activated once. This is clearly a lower bound for $\broadcasttime(G)$ as broadcast requires each agent to be active (at least) once. Thus, the time to sample all required random bits is $O(\broadcasttime(G) \log n)$
Once sampling is complete, the
$\bclevel{}$ and $\bcgroup{}$ fields converge to the first
missing value $(M_L, M_G)$. The proof constructs a chain of \emph{witness
agents} $v_1, \ldots, v_k$ whose $(\level{}, \group{})$-tuples
cover every predecessor of $(M_L, M_G)$ in lexicographical order. Convergence
proceeds as a sequence of broadcasts.
\begin{enumerate}
\item Agent~$v_1$'s value propagates through the graph until it reaches agent~$v_2$.
\item Contact with~$v_2$ triggers the next lexicographic increment and a new epidemic.
\item This cascades through $v_3, \ldots, v_k$ and finally propagates $(M_L, M_G)$ to all agents.
\end{enumerate}
Since $k = O(\log^2 n)$ w.h.p.\ (there are at most $2\ell$ groups per level~$\ell$, and $M_L = O(\log n)$), the total time is bounded by $O(\log^2 n)$ consecutive broadcast times, yielding $O(\log^2 n \cdot \broadcasttime(G))$.

\subsection{Counting the Exact Population Size}
\label{sec:exactcountingshort}
\enlargethispage{2\baselineskip}
\long\def\descriptionCounting{
To count the exact number of nodes we follow the general idea by \textcite{DBLP:conf/podc/BerenbrinkKR19}.
Initially, the leader starts with $M$ discrete tokens.
These tokens are balanced using a classical load balancing algorithm \cite{sodaPaper}.
Let $L_u(t)$ denote the load of node $u$ at time~$t$.
Once the load balancing has converged, all nodes have a load of $M / n \pm 2$.
Each agent outputs at time $t$ the value
$\operatorname{round}(M / L_u(t))$ where $\operatorname{round}(\cdot)$ is a function rounding to the nearest integer. Note that $L_u(t)$ is not necessarily exactly the same for all of them (which it cannot be when $n$ does not divide $M$).
To make sure that $\operatorname{round}(M/L_u(t))$ still returns the same value for all agents $M$ has to be by a factor $n$ larger than the number of agents. We ensure this by choosing an $M = \Omega(n^2)$.

Fortunately, to choose $M$ appropriately we do not have to know $n$ exactly.
Our protocol first computes an approximation of $n$ (called $\bclevel{}$) using \EstLogN\ and then computes $M$ as a function of that estimate.
That way, once $\EstLogN$ has converged, all agents can calculate a correct value for $M$, as required for correctness.
This allows us to store just the (integer-valued) binary logarithm of the number of red tokens in a field called $\logred{}$, requiring only logarithmically many states. The field will store the value $\bot$ if an agent does not store any red tokens.

In our protocol the leader initially creates $2^{\ell}$ red tokens ($\ell$ has to be chosen suitably and we calculate it later). Each red token is worth $M/2^{\ell}$ blue tokens. Our protocol first roughly balances the number of red tokens as follows:
When an agent with red tokens interacts with an agent without any red tokens, the former sends half of its tokens to the latter.
If both agents have red tokens, then the agents swap their tokens with probability 1/2.
This ensures two things. Firstly, the number of red tokens stored at any agent is always a power of two. Secondly, when viewed individually, each red token performs a random walk.

Whenever an agent stores only one red token (and not too many blue tokens already) the agent converts the red token into $M/2^{\ell}$ blue tokens. Blue tokens are balanced in the standard way using the algorithm from \cite{sodaPaper}: when two agents interact, they balance their load as much as possible, with any round-off remaining at the initiator. These blue tokens then balance up to a discrepancy of at most three.

\medskip

The combination of our red tokens representing powers of two and the maximum number of blue tokens allowed per agent creates significant probabilistic dependencies between the movement of red and blue tokens. The main problem in our analysis is that agents with red tokens are only allowed to balance with agents having no red tokens. Since agents are only allowed to balance the load with their neighbors it is clear that the number of red tokens on nodes with a small distance in the graph is highly correlated. For this reason we swap the tokens in interactions where both agents have already red tokens. Note that these dependencies complicate the analysis of our protocol considerably.

\subsection{Formal Description of our Protocol}

Our counting protocol consists of two parts, the outer protocol \CountingProt(u,v)\ and the subroutine \Balance(u,v).
We assume that there is a unique agent $l$ for which the field $\isleader{l} = \ttrue$, and $\isleader{v} = \tfalse$ for all other agents $v \neq l$.
In addition to the fields of $\EstLogN$,
agent $v$ has fields $\logred v$ and $\blue v$ storing the logarithm of the number of red tokens and the number of blue tokens, respectively.

Initially, the leader starts with $2^{r(\bclevel{})}$ red tokens,
and red tokens are converted into $2^{b(\bclevel{})}$ blue tokens,
where $r$ and $b$ are given by
\begin{align} \hidetag\label{eq:b}
r(x) = \max\{0, \lfloor x - 12 + \log(x - 8) \rfloor\} \quad \text{ and } \quad
b(x) = \lceil 2(x + 2 \log(x) + 7)\rceil + 3 - r(x).
\end{align}
As initially we have $\bclevel{}=1$ for all agents,
the leader initializes $\logred{l}$ to $r(1)$ (which is zero) and all other agents initialize it to $\bot$.
The field $\blue{}$ is initialized to $0$ for all agents.

Now until $\EstLogN$ has converged, $\bclevel{}$ may change---although it is monotonically non-decreasing (see \cref{{obs:estlogn_props}}).
Whenever $\bclevel{}$ is increased, all agents reset their loads and start all over again.
More precisely, all nodes discard their red and blue tokens, and the leader $l$ restarts with $\logred{l}$ set to $r(\bclevel{l})$.
This ensures that agents only balance their load with other agents having the same $\bclevel{}$ value, preserving the total load in the system.
All this is done in the outer protocol $\CountingProt$.

\Balance(u,v) first balances the red tokens, converts red tokens into blue tokens and balances the blue tokens. If a node $u$ holds red tokens, it can only convert a red token into blue tokens if it has at most
$ 2 \cdot 2^{b(\bclevel u)} + 3$ blue tokens.
One red token converts into $2^{b(\bclevel u)}$ blue tokens, where $b$ is the function defined above.
For our exact counting protocol we define the output function for a node $u$ as $\rho(\bclevel{u},\blue{u}) = \operatorname{round}\left(2^{r(\bclevel{u})}\cdot2^{b(\bclevel{u})} / \blue{u}\right)$. This expression is the total load in the system divided by the load of node $u$ (both in terms of blue tokens).

}

\def\hidetag{}
{
\let\subsection\paragraph
\descriptionCounting

}

\begin{lstlisting}[caption={\CountingProt(u, v)}, output={$\rho(\bclevel{u},\blue{u}) = \operatorname{round}\left(2^{r(\bclevel{u})}\cdot2^{b(\bclevel{u})} / \blue{u}\right)$}, label={alg:counting}]
$\EstLogN(u, v)$

for $i \in \{u, v\}$ do
    if $\bclevel i$ has increased then
        if $\isleader i$ then
            $(\logred i,\, \blue i) \gets (r(\bclevel i),\, 0)\label{alg:count-exact:5}$
        else
            $(\logred i,\, \blue i) \gets (\bot,\, 0)\label{alg:count-exact:7}$

$\Balance(u, v)\label{alg:count-exact:8}$
\end{lstlisting}

\enlargethispage{2\baselineskip}

{\let\pleasefill\hfill\def\hfill{}
\begin{lstlisting}[caption={\Balance(u,v)},label={alg:balance}]
/*balance red tokens\hfill\mbox{}*/
if $\logred u \geq 1$ and $\logred v = \bot$ then $\label{alg:balance:2}$
    $\logred u,\, \logred v \gets \logred u - 1$
    
else if $\logred u = \bot$ and $\logred v \geq 1$ then
    $\logred u,\, \logred v \gets \logred v - 1$ $\label{alg:balance:5}$
    
else if $\logred u \leq \logred v$
    $(\logred u,\, \logred v) \gets (\logred v,\, \logred u\label{alg:balance:7}$)

/*balance blue tokens\hfill\mbox{}*/

$\displaystyle (\blue u, \blue v) \gets \smash{\left(\left\lceil \frac{\blue u + \blue v}2 \right\rceil, \left\lfloor \frac{\blue u + \blue v}2 \right\rfloor\right)}\label{ln:balanceblue}$

/*load conversion\pleasefill (See \cref{eq:b} for the definition of $b(x)$.)*/
for $i \in \{u, v\}$ do
    if $\logred i = 0$ and $\blue i \leq 2 \cdot 2^{b(\bclevel i)} + 3$ then $\label{ln:explosion-condition}$
        $\left(\logred i, \blue i\right) \gets \left(\bot, \blue i + 2^{b(\bclevel i)}\right)\label{ln:explosion}$
\end{lstlisting}
}

\begin{theorem}[restate=restateThmMainCounting,label={thm:main_counting}]
Assume $\CountingProt$ is started on a graph $G =(V, E)$ with $n = |V|$ nodes and exactly one of the nodes is marked as leader.
The protocol uses $\tilde{O}(n)$ states, w.h.p.
$\CountingProt$ stabilizes w.h.p.\ in \[O( \broadcasttime(G) \cdot \log^2(n) + \loadbalancingtime(G) \cdot \log(n))\] interactions, and at that time each agent outputs the correct value of $n$ w.h.p.
\end{theorem}

Note that the worst-case stabilization time across all graph classes is $O(n^4 \log^2 n)$.
This time is achieved, e.g., on the lollipop graph that consists of a path connected to a clique, each containing $n/2$ agents.
However, for many graph classes the stabilization time is much lower.
In \cref{tab:overview_main_counting} we present the performance of our protocol for various important graph classes like regular graphs and $d$-dimensional tori.
Therein, we also compare the stabilization time of our protocol with the stabilization time of the baseline protocol (see \cref{sec:folklore}) as a frame of reference.
In particular, we note that for regular expanders the stabilization time is almost linear, which is optimal within logarithmic factors, and improves upon the baseline protocol almost by a linear factor.

\begin{table}
\caption{Comparison between the stabilization times of our new protocol and the baseline protocol on several graph classes. $\lambda$ denotes the second-largest eigenvalue of the adjacency matrix of $G$. Note that for the Clique, \cite{DBLP:conf/podc/BerenbrinkKR19} is faster than both. Note that $(1-\lambda(\mP(G)))$ denotes the spectral gap.}
\label{tab:overview_main_counting}

\centering
\begin{tabularx}{\textwidth}{llllX}
\toprule
\bfseries Graph Class & \bfseries $O(\broadcasttime(G)\log^2 n + \loadbalancingtime(G)\log n)$ & \bfseries $O(\hittingtime(G)\cdot n \cdot \log n)$ & \bfseries Comment\\
\midrule
Clique & $O(n\log^3 n)$ & $O(n^2\log n)$ \\
Regular graphs & $O(\frac{n}{1-\lambda(\mP(G))}\log^3 n)$ & $O(\frac{n^2}{\sqrt{1-\lambda(\mP(G))}}\log n)$ & $\frac{1}{1-\lambda(\mP(G))} \in O(n^2)$ \\
Cycle & $O(n^3 \log^2 n)$ & $O(n^3 \log n)$\\
$d$-dim. Torus ($d=2$) & $O(n^2\log^2 n)$ & $O(n^{2}\log^2 n)$\\
$d$-dim. Torus ($d>2$) & $O(n \cdot n^{2/d} \cdot \log^2 n)$ & $O(n^{2}\log n)$\\
Worst Case & $O(n^{4}\log^2 n)$ & $O(n^{4}\log n)$ & Tight for Lollipop\\
\bottomrule
\end{tabularx}
\end{table}

\paragraph{Proof Sketch for \cref{thm:main_counting}} The bound on the number of states is straightforward. In our proof we assume that protocol $\EstLogN$ stabilizes such that all nodes have the same value of $\bclevel{}$. We call this time $T_0$.
The rest of the proof is split into three parts.

\begin{itemize}

\item \emph{Part 1.} Let $T_1$ be the first time when all agents have at most one red token such that $\logred{u} \leq 0$ for any node $u$.
\Cref{lem:red-balance-time} shows that $T_1 - T_0 = O(\log(n) \cdot \loadbalancingtime(G))$ w.h.p.

\item \emph{Part 2.} Let $T_2$ be the first time when all red tokens have been converted into blue tokens such that $\logred{u} = \bot\ \forall u$.
\Cref{lem:eliminating-red-tokens} shows that $T_2 - T_1 = O\left(\log(n) \cdot \loadbalancingtime(G)\right)$ w.h.p.

\item \emph{Part 3.} Let $T_3$ be the first time when all blue tokens are balanced up to discrepancy $3$ such that $\max_u\{\blue u\} - \min_v \{ \blue v \} \leq 3$. Then
\Cref{lem:blue-balance-time} shows that \mbox{that $T_3 - T_2 = O\left(\log(n) \cdot \loadbalancingtime(G)\right)$ w.h.p.}
\end{itemize}
The bound on the running time follows by taking the union bound over these parts. Our algorithm outputs the correct size of the population since the number of blue tokens is at least $\Omega(n^2)$. Hence, we can apply the same analysis as in \cite{DBLP:conf/podc/BerenbrinkKR19}. It remains to give some more details about the three parts of our proof. Since the proof of the third part follows immediately from a known load balancing result \cite{sodaPaper} we discuss here only the first two parts.

\paragraph{Part 1} This is the most challenging of our proofs.
The main idea of the proof is to track the positions of individual red tokens and to bound the time until the token is the sole red token on its host node. In our proof we show that every red token, when considered by itself, performs a lazy random walk on the graph. Note that initially all red tokens are on the same node. Tokens that are on the same node travel together for a while until they separate due to a halving step. Recall that a halving step occurs exactly when a node with at least two red tokens interacts with a node without red tokens (see lines \ref{alg:balance:2} to \ref{alg:balance:5} in \cref{alg:balance} $\Balance$). Until then their movements are highly correlated and, in fact, identical.

To prove the theorem we fix a red token $i$ and consider disjoint epochs of length $O(\loadbalancingtime(G))$.
Now we consider, for each epoch $j$, the last interaction of token $i$ in epoch $j$. Our goal is to show that there is a constant probability of $i$ being involved in a halving step in exactly that interaction (unless, of course, it has already halved earlier in epoch $j$).
Since in the beginning (as long as $\EstLogN{}$ stabilized on a correct value)
there are at most $n/c$ red tokens (for a constant $c$),
$i$ must be alone after $\log_2 n$ halving steps.

It remains to show that for token $i$ the probability of a halving step is constant per epoch.
A halving step occurs exactly when the node holding token $i$
interacts with a node not having any red tokens.
As there is only a constant fraction of $n/c$ red tokens, there must be many such nodes without red tokens. If halving steps only occurred at the end of an epoch, the stacks of red tokens would be distributed (more or less) uniformly among the nodes of the graph by the end of the epoch. Hence, it would be straightforward to show that interacting with an empty node would happen with constant probability, and we would be done.
The usual method to deal with such situations
is to simply assume that, since there are at most $n/c$ red tokens, the worst-case probability to interact during the last step of an epoch with an empty node is still $1-1/c$. However, stacks of red tokens can be halved at any time during an epoch, and we cannot focus on the last step due to significant dependencies between red tokens which started on the same node at the beginning of the epoch.

\smallskip

We address the problem as follows.
Let $I$ be the set of red tokens which are, at the start of epoch~$j$, in the same stack of red tokens (on the same node) as token $i$.
Tokens in $I$ stay together until they undergo a halving step. Tokens in $I$ and tokens not in $I$ do not move completely independently at random since they can never be on the same node, but they move independently otherwise (we make this precise in \cref{lem:tokens_negative_dependence}).
Now let $v$ be the node
interacting with the node holding $i$ in the final interaction of the epoch.
There are three disjoint cases:
$v$ has no tokens at all, $v$ has tokens from $I$,
or $v$ has tokens, but not from $I$.
The first case is straightforward:
there will be a halving step then and there (unless, of course, $i$ is already alone, in which case there must have been a halving step before).
In the second case it is easy to see that there has been a halving step before,
since the tokens in $I$ are no longer all on the same node.
Hence, only the third case is problematic.
In this case we bound the probability of a halving step by using the fact that tokens starting on different nodes behave almost independently (again \cref{lem:tokens_negative_dependence}).
This allows us to bound the expected number of tokens on $v$ not from $I$ by a constant less than $1$.

\paragraph{Part 2}
At time $T_1$ each node holds at most one red token. We call the nodes holding a red token \emph{red nodes}.
On a high level, the proof proceeds by dividing the execution into non-overlapping epochs of length $2 \cdot \loadbalancingtime(G)$ and showing that in every other epoch the number of red nodes decreases by a constant factor with constant probability.
Recall that a red node will be able to convert its red token into blue tokens as soon as it interacts with a node holding a small enough number of blue tokens.
Fortunately, the number of nodes that can be used for the conversion is sufficiently large as agents constantly balance their blue tokens.

More precisely, we partition time into epochs of length~$2\cdot \loadbalancingtime(G)$.
For $i\ge 0$ let $t_i$ denote the first step of epoch~$i$, and let $R(i)$ denote the number of red nodes at time~$t_i$.
For our proof it will be necessary to distinguish between the blue
tokens a node has at the beginning of an epoch and the blue tokens it receives due to conversions in that epoch.
Hence, we refer to the tokens which are on a node at the beginning of an epoch as \emph{old blue tokens}.
All blue token added in epoch $i$ are \emph{new} blue tokens, respectively.
In the following we denote the set of old blue tokens at the beginning of epoch $i$ as $B'(i)$.
Note that we can assume w.l.o.g.\ that our protocol balances the old blue tokens and the \emph{new} blue tokens independently.
By \cite[Theorem 2.3]{sodaPaper}, the old blue tokens from epoch $i$ are balanced at the beginning of epoch $i+1$ with at least constant probability,
and the discrepancy is at most~$3$.
Since the total number of blue tokens generated by \cref{alg:balance} ($\Balance$) is bounded by $\BN$
there are at most $\BN$ old blue tokens.
Hence, if the tokens are successfully balanced, each node holds at most $\RedToBlue+3$ old blue tokens at the end of each epoch.
Let now $B(i+1)$ be the total number of blue tokens and $\partial B(i+1) = B(i+1) - B'(i)$ be the number of \emph{new} blue tokens created in epoch $i$. We now distinguish between two cases:
\begin{enumerate}
\item If $\partial B(i+1) \leq \frac{2}{3} R(i) \RedToBlue$ then
$\tfrac{1}{2} R(i)$ red nodes have fewer than $2\cdot \RedToBlue+3$ blue tokens, and their red token will be converted into blue tokens in their next interaction (line \ref{ln:explosion-condition} of protocol $\Balance$). W.h.p., this interaction happens in epoch $i+1$.
\item Otherwise, $\partial B(i+1) > \frac{2}{3} R(i) \RedToBlue$. In this case $\frac{2}{3} R(i)$ red nodes must have already converted their token in epoch $i$.
\end{enumerate}
Therefore, given that the old blue tokens are balanced, the number of red nodes decreases by a constant factor, either in epoch $i$ (Case 2) or epoch $i+1$ (Case 1).

\subsection{Impossibility Result}
\label{sec:impossibleshort}

\long\def\ImpossibilityIntuition{

Our protocol converges to the correct estimate of $n$ w.h.p.
However, the agents themselves cannot detect when this stabilization occurs.
In particular, even if an agent currently stores the correct value of $n$, it cannot determine whether this value is final or may still change in future interactions.
One may therefore ask whether it is possible to make the protocol \emph{detect} that it has terminated.
We follow the definition by \textcite{DBLP:conf/podc/DotyE19} and augment each node $u$ with a Boolean variable $\terminated u$ that indicates termination. Ideally, this variable would be set to $\ttrue$ once the agent has computed the correct value of $n$, and remain $\ttrue$ thereafter, at least with some non-trivial probability.
As it turns out, however, this is impossible in our model. More precisely, if the protocol has no a priori knowledge about the underlying graph, then no agent can reliably detect that the correct value of $n$ has been computed, even in the presence of a unique leader.

}

\ImpossibilityIntuition

\begin{theorem}[restate=restateImpossibleThm,label={thm:impossible}]
Let $P$ be a graph-uniform population protocol that computes the size of an interaction graph.
Suppose there is a graph $H$ on which $P$ terminates correctly with positive probability in bounded time~$t_H$ for some initial configuration.
\begin{enumerate}
\item If the protocol starts with a unique leader then there is an infinite sequence of graphs $(G_i)_{i \in \N}$ with $|V(G_i)| < |V(G_{i+1})|$ for all $i$
so that the probability of $P$ terminating incorrectly on $G_i$ is constant.
\item Without a unique leader there is an infinite sequence of graphs $(G'_i)_{i \in \N}$ with $|V(G'_i)| < |V(G'_{i+1})|$ for all $i$
so that the probability of $P$ terminating incorrectly on $G_i'$ is $1-\exp(-\Theta(|V(G_i')|))$.
\end{enumerate}
\end{theorem}

\paragraph{Proof Sketch for \cref{thm:impossible}}
We begin with the first statement. The key idea is to embed the graph $H$ inside a much larger graph $G_i$. The construction connects $H$ as an induced subgraph to the rest of $G_i$ by just a single edge. During the first $t_H$ interactions which involve nodes in $H$, that connecting edge is never used with constant probability.
When this happens, $H$'s nodes cannot tell they are part of a larger graph; they terminate thinking that the graph size is $|H|$.
But the actual size is $|G_i| > |H|$, so they output the wrong answer.
The second statement uses more or less the same setup except all agents start in an identical state. Now the probability of correct termination becomes exponentially small in the graph size. Instead of embedding $H$ once, we embed many disjoint copies of $H$.
We create roughly $n_i/|H|$ copies of $H$, each connected to a central clique by one edge, where $n_i$ denotes the size of the graph with all the copies of $H$.
Each copy independently has constant probability of getting isolated and terminating incorrectly. For the protocol to be correct on $G_i$, \emph{all} copies must avoid this trap.
The probability of success is then $\exp(-\Theta(n_i))$.

\subsection{Generating Independent Random Bits}\label{sec:random-bits-overview}

\long\def\descriptionGeneratingRandomBits{
Our final result is a uniform protocol for generating independent, fair, unbiased random bits.
These bits will be used as a basic primitive in our counting protocol.
While it is not surprising that such a protocol exists, we are not aware of a rigorous analysis; thus, we provide one here.

Just as all other protocols (we are aware of) for sampling bits in graphical population protocols \cite{DBLP:journals/dc/AlistarhRV25, DBLP:conf/opodis/AlistarhGR21}, our protocol derives randomness from the orientation of the edge selected by the scheduler in each step.
Instead of using static coloring as in \cite{kanaya_et_al:LIPIcs.OPODIS.2024.37} to decide which of the two interaction partners may use the random bit implied by the direction, our protocol uses dynamic $2$-coloring to produce a non-monochromatic edge with constant probability.

In a nutshell, the protocol works as follows.
Each node $u$ maintains a color, either black or white, stored in the field $\Color u$.
Whenever nodes $u$ and $v$ interact and have different colors, the black node generates a random bit.
Suppose that $u$ is black.
Then the bit is stored in the field $\Rand u$, taking value $1$ if $u$ is the initiator, i.e., if the scheduler activated edge $(u,v)$, and value $0$ otherwise.
The white node $v$ sets $\Rand v$ to $\bot$, meaning that it created no random bit in this interaction.
If both endpoints have the same color, the corresponding values $\Rand u$ and $\Rand v$ are both set to $\bot$.
At the end of the step, the colors of $u$ and $v$ are updated to black and white, respectively.}
\descriptionGeneratingRandomBits

\begin{lstlisting}[caption={\GenerateRandomBit($u$, $v$)}, label={alg:random-bits}]
if $(\Color u,\Color v) = (\black,\white)$ then
    $(\Rand u, \Rand v) \gets (1, \bot)$ /*The black node $u$ is initiator, assign bit $1$ */
else if $(\Color u,\Color v) = (\white,\black)$ then
    $(\Rand u, \Rand v) \gets (\bot, 0)$  /*The black node $v$ is responder, assign bit $0$ */
else 
    $(\Rand u, \Rand v) \gets (\bot, \bot)$
$ ( \Color u, \Color v ) \gets (\black, \white) \label{alg:random-bits-ln-7}$
\end{lstlisting}

\begin{theorem}[restate=restateBitSampling,label={thm:bitsampling}]
Let $T$ be an arbitrary step and assume that node $u$ interacts in steps $T_u \subseteq [T]$.
For $t_j \in T_u$, let $R_u(t_j)$ be the value of $\Rand u \in \{0, 1, \bot\}$ at the end of the interaction in step $t_j$.
For node $u$, define $R'_u(T) = \{ R_u(t_i) \mid R_u(t_i) \in \{0,1\}, t_i \in T_u \}$. 
\begin{enumerate}
\item For all $R(t_i) \in {R'_u}(T)$ we have $\pr{{R}'_u(t_i) = 1} = \pr{R'_u(t_i) = 0} = \tfrac{1}{2}$.
\item All variables in $\bigcup_{v\in V} R'_{v}(T)$ are mutually independent.
\item For any step $t \geq 1$, each node $u$ generates a random bit in steps $[t, t + {4|E(G)|}/{\delta(G)}]$ with probability at least ${1}/{96}$. This probability is independent for non-overlapping time intervals.
\end{enumerate}
\end{theorem}
The first two statements show that our protocol indeed generates uniformly random and independent bits. The third result states that each node $u$ creates a random bit roughly every $|E(G)|/\delta(G)$ steps where $\delta(G)$ is the minimal degree of $G$. Note that the number of random bits a node generates is proportional to the number of times it interacts.

We prove the theorem in \cref{sec:randombits}.
Moreover, we derive a user-friendly statement about the number of bits that are sampled within a certain time frame.
Using well-known tail estimates, the following corollary directly follows from the third property in the above theorem.
\begin{corollary}[restate=restatebitsamplingRepeat,label={cor:bitsamplingRepeat}]
Let $\ell \geq 1$ be an integer.
\begin{enumerate}
\item In $O\left(\tfrac{|E(G)|}{\delta(G)} \cdot \ell\right)$ steps, a node $u$ draws at least $\ell$ random bits with probability at least $1-2^{-\ell}$.
\item For $\ell \geq (c+1)\log n$, in $O\left(\tfrac{|E(G)|}{\delta(G)} \cdot \ell\right)$ steps, \emph{each} node draws at least $\ell$ random bits with probability at least $1-n^{-c}$.
\end{enumerate}
\end{corollary}

\paragraph{Proof Sketch for \cref{thm:bitsampling}} For this protocol, it is relatively easy to show that each black node receives the bit $0$ and $1$ with probability $1/2$ each.
This holds since the random edge orientation is assigned to only one of the two activated nodes in each step. Showing that each node receives random bits regularly is harder.
To show this we partition time into phases of length $\Theta\left({|E(G)|}/{\delta(G)}\right)$.
Fix an arbitrary phase $i$ and fix all random choices up to this phase.
After this conditioning, all remaining randomness stems solely from the scheduler's random choices within phase~$i$.
Recall that a random bit is created by node $u$ whenever node $u$ is black and interacts with a white node $v$. In turn the
colors of $u$ and $v$ are determined by the outcome of their previous interactions. Since we have an arbitrary coloring of $G$ at the beginning of the phase, $u$ possibly starts with the same color as all its neighbors, e.g., white.

Now the length of the phase comes into play.
Using elementary calculations, we show that with constant probability the following happens in phase $i$:
node $u$ interacts with a node $w$, and both $u$ and $w$ had another interaction (not necessarily with each other) earlier in phase $i$.
Thus, in this case, the color of $u$ and $w$ is only determined by the random choices during these interactions, thus it is independent of the graph coloring at the beginning of the phase:
\begin{itemize}
\item The most recent interactions of $u$ and $w$ were with each other. In this case they have different colors, and $u$ is black with probability ${1}/{2}$ (see Line \ref{alg:random-bits-ln-7} of \cref{alg:random-bits}).
\item The most recent interactions of $u$ and $w$ occurred in different interactions.
In this case the color of $u$ and $w$ is determined exclusively by the orientation of the edge in these interactions (again, Line \ref{alg:random-bits-ln-7} in \cref{alg:random-bits}). Hence, the probability that $u$ is black and $v$ is white is ${1}/{4}$.
\end{itemize}
It follows that, considering the second interaction of node $u$, the colors of nodes $u$ and $w$ are either deterministically different or they are independent random variables.
Note that the interaction pattern only depends on the pairs of nodes that interact, while the coloring depends on the orientation of the edges.
As these are chosen independently of each other, the probability of having the interaction pattern and the appropriate coloring is constant. We give the detailed proof in \cref{sec:randombits}.

\subsection{Summary and Future Work}

We present a fast population protocol on graphs to count the exact number of nodes.
Our protocol requires no knowledge of the underlying graph.
Its performance depends only on spectral properties of the random scheduler, the degree distribution of $G$, and the broadcast time $\broadcasttime(G)$.

Our approach separates the task into three independent components: generation of independent random bits on graphs, a tight approximation of $\log n$ via a refined first-missing-value construction, and load-balancing-based exact counting driven by a unique leader.
We believe that the first two components are of independent interest.
The resulting protocol requires $\tilde O(n)$ states and stabilizes w.h.p.\ in a time proportional to the broadcast time $\broadcasttime(G)$ and the load balancing $\loadbalancingtime(G)$ time of the graph.
Within logarithmic factors, our approach is always as fast as the baseline protocol based on coalescing random walks.
For important graph classes like regular expanders, i.e., graphs with constant edge expansion, or $d$-dimensional torii (with $d > 3$), our new protocol is faster.
We conjecture that the protocol can be turned into an \emph{always-correct} protocol which is guaranteed to eventually output the correct value of $n$ with probability $1$, using the same convergence time and state complexity.
Finally, we present an impossibility result for terminating uniform protocols that compute strongly increasing size properties, both with and without a leader.

Further, our protocol is state-optimal within logarithmic factors.
To encode the exact number of nodes, every node needs at least $\Omega(n)$ states.
Our protocol, on the other hand, requires $\tilde{O}(n) = n \cdot \polylog n$ states.
For the sake of presentation, we did not attempt to optimize the degree of the polylogarithmic term in the expression for the number of states.

Finally, we discuss the assumption that our protocol is initialized with a unique leader. Although this is a genuine restriction, unique leaders are a well-established resource in the population-protocol literature. For example, Angluin, Aspnes, and Eisenstat \cite{DBLP:journals/dc/AngluinAE08a} study fast computation with a given leader, while the approximate- and exact-counting protocols of Berenbrink, Kaaser, and Radzik \cite{DBLP:conf/podc/BerenbrinkKR19} elect and subsequently rely on one. 
On the other hand, there exists simple graph-uniform protocol that can count exactly without a leader using $O(n)$ states; however, its stabilization time is governed by the maximal hitting time of a random walk on the underlying graph and can, for many graph families, be far larger than $\widetilde{O}(\broadcasttime(G)+\loadbalancingtime(G))$. 
Thus, the leader does not trivialize the problem but enables a substantially faster protocol whose construction and analysis remain technically involved.

Nevertheless, eliminating the leader assumption while keeping a near-linear state complexity and a low running time is an important direction for future work. The closest known result is the graph-uniform leader-election protocol of Alistarh, Rybicki, and Voitovych \cite{DBLP:journals/dc/AlistarhRV25}, which requires no prior knowledge of the graph and stabilizes in $O(\broadcasttime(G)+n\log n)$ interactions, but uses $O(n^4)$ states. This leaves open the concrete problem of designing a stabilizing leader-election protocol for arbitrary graphs that is graph-uniform, uses $\widetilde{O}(n)$ states, and stabilizes in $\widetilde{O}(\broadcasttime(G)+\loadbalancingtime(G))$ interactions.

\section{Approximate Counting}

\label{sec:approxcouinglong}

In this section we present all technical details of our algorithm for approximate counting.
First each agent $v$ samples a value $\level v$ having geometric distribution $\Geom(1/2)$.
We aim to determine the smallest level value where less than roughly $\log(n)$ agents have said value.
In lieu of a space-efficient and fast procedure to do this directly, we instead let agents additionally sample a $\group v$ ranging from $0$ to $2^{\lceil \log(\level v) \rceil} - 1$ (i.e., a non-negative integer whose binary representation has at most $\lceil \log(n) \rceil$ bits) chosen u.a.r.
If for a given level $L$ all groups have at least one agent, then we presume that this level has still too much support. This is reasonable to assume from a coupon collection argument: For each level $L$ there are between $L$ and $2L$ groups belonging to that level.
Each agent which sampled level $L$ (there are expectedly $n/2^{L}$ of these agents) samples one group u.a.r.
If all groups are chosen by at least one agent we conclude that
the number agents on that level was sufficiently large to sample the full set.
It is well known that in expectation $\Theta(L \log(L))$ agents are needed to collect all groups.
Since we are looking for a level of order $\log n$,
this works, as we show below. If we now assume that level/group tuples are ordered lexicographically, we look for the first tuple which was not sampled and output the level of that group.
Hence, we are looking for the first missing value of the tuples.

\subsection{Formal Description}\label{approx_alg}

The algorithm consists of two parts. First, agent $v$ use the random bits sampled by
\cref{alg:random-bits} to determine its random variable $\level v$.
After that the agent samples its random variable $\group v$ (whose range is a function of $\level v$), again using \cref{alg:random-bits}.
Finally that agents are looking for the first missing value in the lexicographic order.
We first describe the sampling of $\level u$ and $\group u$ and then the computation of the first missing value.

\paragraph{Sampling Level and Group}

$\Sample$ samples the $\level u$ and $\group u$ values.
To control its behavior, each agent has two Boolean fields:
$\issampling{}$, initialized to $\ttrue$, indicates whether sampling is ongoing,
and $\leveldone{}$, initialized to $\tfalse$, indicates whether sampling the $\level{}$ value has finished (so that the protocol may proceed sampling the $\group{}$ value).
$\Sample$ first generates a random bit using \cref{alg:random-bits} from the previous section.
If an agent has a bit available (meaning that $\rand{} \neq \bot$),
it uses it to build up $\level{}$ and $\group{}$ values as follows.
$\level{}$ is initialized to $1$
and then incremented for each coin flip showing $1$;
if the coin shows $0$, sampling the $\level{}$ is done.
The final value of $\level{}$ is therefore the number of fair coin flips until the first zero is observed. It has distribution $\Geom(1/2)$ as desired.

The possible values for $\group{w}$, from which we want to sample u.a.r., are exactly those integers whose binary representation have at most $\lceil \log \level{w} \rceil$ bits.
Their binary representation has the distribution
of a sequence of $\level{w}$ bits sampled independently and uniformly at random.
The protocol samples $\group{w}$ in exactly this way.

\begin{lstlisting}[caption={\Sample(u, v)}, label={alg:SampleLevels}]
$\GenerateRandomBit(u, v)$

for $w \in \{u, v\}$ with $\rand w \neq \bot$ and $\issampling w = \ttrue$ do
    if $\neg \leveldone w$ then
        if $\rand w = 1$ then $\level w \gets \level w + 1$
        else $\leveldone w \gets \ttrue$
    else 
        $\group w \gets \group w + 2^{\groupbitssampled w} \cdot \Rand w$
        $\groupbitssampled w \gets \groupbitssampled w + 1$
        if $\groupbitssampled w = \lceil \log \level w \rceil$ then
            $\issampling w \gets \tfalse$
\end{lstlisting}

\paragraph{Coupon Collection}

We consider the lexicographic order for set $\tupleset = \{(l,g) \in \N_+ \times \N_0 \mid g < 2^{\lceil \log(l)\rceil}\}$ of valid $(\level v, \group v)$ tuples.
In this order, we have for two tuples $(l_1, g_1), (l_2, g_2) \in \tupleset$ that $(l_1, g_1) \leq_\lex (l_2, g_2) \Longleftrightarrow (l_1 < l_2) \vee (l_1 = l_2 \wedge g_1 \leq g_2)$.
This is a total order on $\tupleset$, so that we can define $\max_\lex$ and $\min_\lex$ in the usual way.
Furthermore, let $\succlex((l, g)) = \min_\lex \{(l', g') \in \tupleset \mid (l', g') >_\lex (l, g)\}$ be the successor function for this order, which is well-defined and computable as
\[\succlex((l, g)) = \begin{cases}
(l+1, 0),\quad\textup{If $g = 2^{\lceil \log(l)\rceil} - 1$,} \\
(l, g+1),\quad\textup{otherwise.}
\end{cases}\]

We determine the \emph{first missing value} among agents' tuples under this order:
The first missing value of a set $S \subseteq \mathcal{S}$ is
\[\FMV_\lex(S) = \min_\lex\{(l, g) \in \mathcal{S} \mid (l, g) \not\in S\},\]
so that the value we compute is $\FMV(\{(\level i, \group i) \mid i \in [n]\})$.
The first missing value among tuples will thus have as its first entry the lowest level where a group is missing, and as its second entry the numerically smallest missing group.

The protocol $\EstLogN$ computes this first missing value via broadcasts:
In the protocol $\EstLogN$, agents have fields $\bclevel v$ and $\bcgroup v$ which are both initialized to $1$ and $0$
(so that the initial value of $(\bclevel v, \bcgroup v)$ is $\min_\lex(\tupleset) = (1, 0)$).
Whenever agent $v$ has $(\level v, \group v)=(\bclevel{}, \bcgroup{})$, the agent knows for sure that the tuple $(\bclevel{}, \bcgroup{})$ is \emph{not} missing and thus sets it to the successor tuple, which, as far as it knows, might be missing.
This continues as long as successive tuples are not missing.
Crucially, for the first tuple that is missing, no further increase is possible, so that the protocol will converge.

\begin{lstlisting}[caption={\Cref{alg:EstLogN}.\hspace{.5em}\EstLogN(u, v), extended version}]
if $\issampling u$ or $\issampling v$ then
    $\Sample(u, v)$
    
else
    w.l.o.g. let $(\bclevel u, \bcgroup u) \geq_{\text{lex}} (\bclevel v, \bcgroup v)$ /*otherwise, exchange $u$ and $v$*/
    
    if $\exists i \in \{u,v\} \colon (\level i, \group i) = (\bclevel u, \bcgroup u)$ then
        $(\bclevel u, \bcgroup u) \gets \succlex(\bclevel u, \bcgroup u)$ /*start looking for the next value*/
    
    $(\bclevel v, \bcgroup v) \gets (\bclevel u, \bcgroup u)$ /*broadcast the maximum */
\end{lstlisting}

\subsection[Proof of the Theorem]{Proof of \Cref{thm:approximate-counting}}
The analysis in this section is split into two parts. First, we analyze the running time and the space requirements of our protocols. Then we determine the approximation quality of our estimation. We show the following theorem.

\restateThmApproximateCounting*

For die-hard 
readers we present tighter bounds in \cref{lem:fmv_tuple_tailbound}.
The proof of the theorem directly follows from \cref{{lem:estlogn_states_time}} and \cref{obs:estlogn_props}.

\begin{lemma}\label{lem:estlogn_states_time}
Protocol $\EstLogN$ uses $O(\log^4(n))$ states w.h.p.,
and stabilizes after $O(\log^2(n) \cdot \mathbf{B}(G))$ interactions w.h.p.\
\end{lemma}

\begin{proof}
$\GenerateRandomBit$ uses a constant number of states (as all fields used by it have a fixed set of values).
$\EstLogN$ additionally uses some Boolean fields (which have $2$ states each)
and the four fields $\level{}$, $\group{}$, $\bclevel{}$ and $\bcgroup{}$.
W.h.p., $\max_{i \in [n]}\{\level{i}\} = O(\log(n))$,
and hence also $\max_{i \in [n]}\{\group i\} = O(\log(n))$ (since $\group i < 2\level i$).
Since the $\bclevel{}$ and $\bcgroup{}$ values are at most one larger than the maximum values of $\level{}$ and $\group{}$ among all agents, they are in $O(\log(n))$ as well,
leading to the claimed stat count.


To bound the stabilization time we have to bound the times until agents are done with the sampling of $\level{}$ and $\group{}$ fields,
as well as the time it takes for $\bclevel{}$ and $\bcgroup{}$ to converge.

\paragraph{Sampling} W.h.p.\ the maximum value of $\level u$ is $O(\log(n))$, agent $u$ needs at most $O(\log(n))$ random bits to sample level, w.h.p.
For sampling $\group u$ agent $u$ needs $\lceil \log \level{} \rceil$ random bits. It is easy to see that
\[ \lceil \log \level{} \rceil\in O(\log \log n) = o(\log (n)). \]
From \cref{thm:bitsampling} it follows that
within $O(|E|/\delta(G) \cdot \log(n))$ interactions, agent $v$ will have sampled sufficiently many random bits, with probability at least $1 - n^{-2}$.
Taking a union bound over all agents we see that with a probability at least $1 - n^{-1}$ this is true for all agents within $O(|E|/\delta(G)\cdot \log n)$ interactions, where $\delta(G)$ is the minimum degree from $G$. Now note that the expected time until a node $u$ with degree $\deg(u)$ participates in an interaction is $|E|/deg(u)$. Hence $|E|/\delta(G)$ is a lower bound on the expected broadcast time and $|E|/\delta(G) = O(\broadcasttime(G))$, recalling $\broadcasttime(G)$ is defined as the expected (worst-case) broadcast time of $G$. From this follows that with a probability at least
$1 - n^{-1}$ the running time of $\Sample$ is $O(|E|/|N_v| \cdot \log(n))$.

\paragraph{First Missing Value} Now we bound the time for the $\bclevel u$ and $\bcgroup u$ values to stabilize.
To that end, define $T$ as the time where all agents $u$ have finished sampling $\level u$ and $\group u$ and consider \EstLogN(u, v) from $T$ onward.
By construction, for all $t \geq T$ we have \[(\level u (t), \group u (t)) = (\level u (T), \group u (T)).\]
For ease of presentation we drop $T$ in $\level{}$ and $\group{}$ in the remainder of the proof.
Let \[(M_L, M_G) = \FMV_\lex(\{(\level u, \group u) \mid u \in V\}).\]
By definition
there must be agents $u_1, u_2, \ldots, u_k \in V$
such that \[(\level{u_1},\group {u_1}) = (1, 0),\]
\[ (M_L, M_S) = \succlex((\level{u_k}, \group{u_k})),\]
and for all $j \in [k-1]$ we have \[(\level{u_{j+1}}, \group{u_{j+1}}) = \succlex((\level{u_j}, \group{u_j})).\]

Whenever an agent $u_j$ first interacts with an agent $v$ having \[(\bclevel{v}, \bcgroup{v}) = (\level{u_j}, \group{u_j})\]
the agents $u_j$ and $v$ both set their $(\bclevel{}, \bcgroup{})$ to $\succlex((\level{u_j}, \group{u_j}))$.

This increased value of $(\bclevel{}, \bcgroup{})$ is now propagated via broadcast through the graph. The propagation might be
preempted if there appears a $(\bclevel{}, \bcgroup{})$-tuples with even larger values.
From this it follows that the time to reach a configuration where all agents have $(\bclevel{}, \bcgroup{}) = (M_L, M_G)$ is bounded by $k$ times the broadcast time.

It remains to bound the value of $k$.
Since $M_L = O(\log(n))$ w.h.p.\ (\cref{lem:fmv_tuple_tailbound}) and there are at most $2L$ groups at level $L$ we have $k= O(\log^2(n))$ w.h.p.
Hence the time in question is bounded by the time of $O(\log^2(n))$ consecutive broadcasts resulting in an expected running time of
$O(\log^2 n \cdot \broadcasttime(G)). $
From simple Markov bounds it follows that with a probability of $1/2$ each of these broadcasts is finished after
time $2 \broadcasttime(G)$, Applying a standard Chernoff bound one can argue that w.h.p.\ the running time of \EstLogN\ is bounded by $O(\log^2 n \cdot \broadcasttime(G))$.
\end{proof}

Next we present the following observation describing the values sampled by our algorithm.

\begin{observation}\label{obs:estlogn_props}
Consider the variables determined by \cref{alg:SampleLevels} and \cref{alg:EstLogN}.
\begin{itemize}
\item For all agents, the value of $\bclevel{v}$ is non-decreasing over time.
\item After executing $\EstLogN(u, v)$, we have $\bclevel{u} = \bclevel{v}$.
\item Assume at time $t$ all agents $v$ have $\issampling{u} = \tfalse$
and
\[(\bclevel{v}, \bcgroup{v}) = \FMV_{\lex}(\{(\level u, \group u) \mid u \in V\}),\]
then $\bclevel{v}$ does not change anymore.
\item For each node $v$ we define $t_v$ as the time when agent $v$ finishes the sampling via the protocol $\Sample$.
Then we have
\[\level v(t_v) \sim \Geom(1/2), \quad \text{ and } \quad \group v(t_v) \sim \UniformDistr\{1, \ldots, 2^{\lceil \log \level v(T_1) \rceil}\}.\]
\end{itemize}
\end{observation}

\paragraph{Approximation Quality}

In this section we estimate the approximation quality of our regime.
For the ease of presentation we use in the lemma $L_i$ instead of $\level i$ and $G_i$ instead of $\group i$.

We will need the following Chernoff bound.
\begin{theorem}[Theorem 7 in \cite{DBLP:journals/ipl/DillencourtGM25}]\label{thm:chernoff_parameterized}
Let $X_1, X_2, \ldots, X_n$ be independent random variables taking values in $\{0, 1\}$.
Let $X = \sum_{i=1}^n X_i$ and let $\mu = \E{X}$ denote $X$'s expected value.
Then \[\Pr[X > R] \leq 2^{-xR}\textup{ for $x > 0$ and $R \geq (2^x e - 1)\mu$.}\]
\end{theorem}

\begin{lemma}\label{lem:fmv_tuple_tailbound}
Let $(L_i)_{i \in [n]}$ be i.i.d.\ $\Geom(1/2)$ r.v.s,
and for all $i \in [n]$ let $G_i \sim \UniformDistr(\{0, \ldots, 2^{\lceil \log L_i \rceil}-1\})$ chosen independently from all other values (except of course $L_i$).
Let $(M_L, M_G) = \FMV_\lex(\{(L_i, G_i) \mid i \in [n]\})$.

\begin{enumerate}
\item \(\Pr[M_L \leq \log(n) - \log \log n + 2 + \log(e) + 2c] \geq 1 - 2 n^{-c}\) for all $c > 0$ and $n \geq \max\{16, 2^{2/c}\}$.
\item \(\Pr[M_L \geq \log(n) - \log \log n - \log(c\ln(n) + 2\ln \log(n))] \geq 1 - 2n^{-c}\) for all $c \geq 1$ and $n \geq 3$.
\end{enumerate}

\end{lemma}

\begin{proof}
\emph{(1)} Let $N_l$ be the number of $i$s for which $L_i = l$, i.e., $N_l = \abs{\{i \in [n] \mid L_i = l\}}$.
If $N_l < l$ (and hence $N_l \leq 2^{\lceil \log(l)\rceil}$) for some level $l$,
there must be a $g \in \{0, \ldots, 2^{\lceil \log L_i \rceil} - 1\}$
such that there is no $i$ with $(L_i, G_i) = (l, g)$,
so that $M_l \leq l$.
Now consider some $c > 0$ and $l = \log(n) - \log\log(n) + 2 + \log(e) + 2c$.
Then $\Pr[L_i=l]=\frac{\log(n)}{n} \cdot \frac{2^{-2c}}{4e}$ and $\E{N_l} = \frac{2^{-2c}}{4e} \log(n)$.
To use the tail bound of \cref{thm:chernoff_parameterized},
see that we hence have $(2^{2c} e - 1) \E{N_l} \leq \frac{1}{4} \log(n) \leq \frac{1}{2}\log(n) - \frac{1}{2c}$ (for $n \geq 2^{2/c}$),
so that
\[\Pr[N_l > \frac{1}{2}\log(n) - \frac{1}{2c}] \leq 2^{-2c(\frac{1}{2}\log(n)-1/(2c))} = 2 n^{-c}.\]
And so as long as $\log\log(n) \leq \frac{1}{2} \log(n)$ (which holds for $n \geq 16$)
we have \[\Pr[M_L > l] \leq \Pr[N_l \geq l] \leq \Pr[N_l > \frac{1}{2}\log n - \frac{1}{2c}] \leq 2 n^{-c}.\]

\emph{(2)}
Let $A_{l,g}$ be $1$ if and only if there is no $i$ with $(L_i, G_i) = (l, g)$, and $0$ otherwise.
Then for all $l$ and $g \in \{0, \ldots, 2^{\lceil \log(l) \rceil} - 1\}$,
\[\E{A_{l,g}} = \Pr[A_{l,g}=1] = \left(1 - \frac{1}{2^l 2^{\lceil \log(l) \rceil}}\right)^n \leq \left(1 - \frac{1}{2^l l}\right)^n \leq \exp\left(-\frac{n}{2^l l}\right).\]
Now for $\hat{l} = \log n - \log \log n - \log (c \ln n + 2\ln \log n)$ (which is at most $\log n$ for $c \geq 1$ as long as $n \geq 3$),
we have
\[\exp\left(-\frac{n}{2^{\hat{l}} \hat{l}}\right) \leq \exp\left(\frac{-n}{\frac{n}{(\log n)(c \ln n + 2 \ln \log n)} \cdot \log n}\right) = \exp\left(-c\ln(n)-2\ln\log n \right) = n^{-c} \log^{-2}(n).\]
Now $M_L \leq l$ if and only if $\sum_{l'=1}^l \sum_{g=0}^{2^{\lceil \log(l) \rceil} - 1} A_{l',g} \geq 1$ (i.e., some level-group-tuple is missing up to that value of $l$).
So by Markov's inequality and linearity of expectation,
\begin{align*}
\Pr[M_L \leq \hat{l}]
&\leq \Pr[\sum_{l'=1}^l \sum_{g=0}^{2^{\lceil \log l' \rceil}-1} A_{l',g} \geq 1]
\leq \E{\sum_{l'=1}^{\hat{l}} \sum_{g=0}^{2^{\lceil \log l' \rceil}-1} A_{l',g}}
\leq \sum_{l'=1}^{\hat{l}} 2^{\lceil \log l' \rceil} \E{A_{l',g}}
\\ &\leq \sum_{l'=1}^{\hat{l}} 2l' \E{A_{\hat{l},g}}
\leq 2 \hat{l}^2 n^{-c} \log^{-2}
\leq 2 \log^2 n \cdot n^{-c} \log^{-2}(n)
= 2n^{-c}. \qedhere
\end{align*}
\end{proof}

\section{Counting the Exact Population Size}

{
\def\hidetag#1#2{\tag{\ref{eq:b}}}
\descriptionCounting
}

{
\def\dcmGobble#1#2{}
\begin{lstlisting}[title={\Cref{alg:counting}.\hspace{.5em}\CountingProt(u, v)},output={$\rho(\bclevel{u},\blue{u}) = \operatorname{round}\left(2^{r(\bclevel{u})}\cdot2^{b(\bclevel{u})} / \blue{u}\right)$}]
$\EstLogN(u, v)$

for $i \in \{u, v\}$ do
    if $\bclevel i$ has increased then
        if $\isleader i$ then
            $(\logred i,\, \blue i) \gets (r(\bclevel i),\, 0)\label{alg:count-exact:5}$
        else
            $(\logred i,\, \blue i) \gets (\bot,\, 0)\label{alg:count-exact:7}$

$\Balance(u, v)\label{alg:count-exact:8}$
\end{lstlisting}
}

{\def\dcmGobble#1#2{}\let\pleasefill\hfill\def\hfill{}
\begin{lstlisting}[title={\Cref{alg:balance}.\hspace{.5em}\Balance(u,v)}]
/*balance red tokens\hfill\mbox{}*/
if $\logred u \geq 1$ and $\logred v = \bot$ then $\label{alg:balance:2}$
    $\logred u,\, \logred v \gets \logred u - 1$
    
else if $\logred u = \bot$ and $\logred v \geq 1$ then
    $\logred u,\, \logred v \gets \logred v - 1$ $\label{alg:balance:5}$
    
else if $\logred u \leq \logred v$
    $(\logred u,\, \logred v) \gets (\logred v,\, \logred u\label{alg:balance:7}$)

/*balance blue tokens\hfill\mbox{}*/

$\displaystyle (\blue u, \blue v) \gets \smash{\left(\left\lceil \frac{\blue u + \blue v}2 \right\rceil, \left\lfloor \frac{\blue u + \blue v}2 \right\rfloor\right)}\label{ln:balanceblue}$

/*load conversion\pleasefill (See \cref{eq:b} for the definition of $b(x)$.)*/
for $i \in \{u, v\}$ do
    if $\logred i = 0$ and $\blue i \leq 2 \cdot 2^{b(\bclevel i)} + 3$ then $\label{ln:explosion-condition}$
        $\left(\logred i, \blue i\right) \gets \left(\bot, \blue i + 2^{b(\bclevel i)}\right)\label{ln:explosion}$
\end{lstlisting}
}

\subsection[Proof of the Theorem]{Proof of \Cref{thm:main_counting}}
For the following analysis assume that $\EstLogN$ has stabilized on the correct value.
Hence, we condition on the high probability event that the statements from \cref{thm:approximate-counting} holds
and prove the following theorem.

\restateThmMainCounting*
\begin{proof}We first show the bound on the running time and the correctness, and then the bound on the number of states.
\paragraph{Running Time and Correctness}
The proof of the running time follows from \cref{thm:approximate-counting,lem:red-balance-time,lem:blue-balance-time,lem:eliminating-red-tokens}.

\begin{itemize}
\item Let $T_0$ be the first time when $\EstLogN$ stabilizes such that all nodes have the same value of $\bclevel{}$.
\Cref{thm:approximate-counting} shows that $T_0 = O(\log^2(n) \cdot \broadcasttime(G))$ w.h.p.

\item Let $T_1$ be the first time when all agents have at most one red token such that $\logred{u} \leq 0$ for any node $u$.
\Cref{lem:red-balance-time} shows that $T_1 - T_0 = O(\log(n) \cdot \loadbalancingtime(G))$ w.h.p.

\item Let $T_2$ be the first time when all red tokens have been converted into blue tokens such that $\logred{u} = \bot$ for any node $u$.
\Cref{lem:eliminating-red-tokens} shows that $T_2 - T_1 = O\left(\log(n) \cdot \loadbalancingtime(G)\right)$ w.h.p.

\item Let $T_3$ be the first time when all blue tokens are balanced up to discrepancy $3$ such that $\max_u\{\blue u\} - \min_v \{ \blue v \} \leq 3$.
\Cref{lem:blue-balance-time} shows that $T_3 - T_2 = O\left(\log(n \cdot \loadbalancingtime(G)\right)$ w.h.p.

\end{itemize}
Taking the union bound over the individual statements, it follows w.h.p.\ that after O$( \broadcasttime(G) \cdot \log^2(n) + \loadbalancingtime(G) \cdot \log(n))$ all nodes $u$ have the same value of $\bclevel{u}$, have $\logred{u} = \bot$, and have $\blue u = M / n \pm 2$, where $M = R \cdot B = 2^{r(\bclevel{u})}\cdot2^{b(\bclevel{u})}$ is the total number of blue tokens in the system.

Observe that the output of node $u$ is $\rho(\bclevel u, \blue u) = \operatorname{round}(M / \blue u)$.
From \cref{obs:number-of-tokens} it follows that $M = \Omega(n^2)$ (where the constant can be made arbitrarily large by adjusting \cref{eq:b}).
This allows us to apply the same analysis as in \cite{DBLP:conf/podc/BerenbrinkKR19}.
From the same calculations as in the proof of Lemma~4.2 from \cite{DBLP:conf/podc/BerenbrinkKR19} it follows that rounding all values in $ \left[\floor{{M}/({\blue u + 2})},~ \ceil{ {M}/({\blue u - 2})}\right]$ gives always the same integer. This integer is precisely the value of $n$.

\paragraph{State Space Complexity} It remains to show that $\CountingProt$ uses $O(n \polylog(n))$ states w.h.p.
To this end, recall that \cref{alg:EstLogN} ($\EstLogN$) uses $O(\log^4(n))$ states w.h.p.\ (\cref{lem:estlogn_states_time}).
$\logred{}$ is either $\bot$ or a non-negative integer with maximum value $r(\bclevel{})$, where $\bclevel{} = O(\log(n))$ w.h.p.\ (see \cref{thm:approximate-counting}).
$\blue{}$ has maximum value $2 \cdot 2^{b(\bclevel{})} + 3 + 2^{b(\bclevel{})} = O(2^{b(\bclevel{})})$ (this follows form our protocol \Balance {})).
To bound this value,
we bound $b(\bclevel{})$ by first considering $r(x)$ for sufficiently large $x$. We have $r(x) = \lfloor x - 12 + \log(x - 8) \rfloor$,
and thus,
\begin{align*}
b(x)
&= \lceil 2(x + 2 \log(x) + 7)\rceil + 3 - \lfloor x - 12 + \log(x - 8) \rfloor
\\ &\leq 2(x + 2\log(x) + 7) + 4 - x + 12 - \log(x - 8)
\\ &= x + 4 \log(x) - \log(x-8) + 30
\\ &\leq x + 3 \log(x) + 32
\end{align*}
with the last inequality holding for $x \geq 32/3$.

At the time where $\EstLogN$ has stabilized (which is the maximum value of $\bclevel{}$ ever encountered as it is non-decreasing by \cref{obs:estlogn_props}),
we know from \cref{lem:fmv_tuple_tailbound} that w.h.p.\ $\bclevel{} \geq \log(n) - O(\log\log(n))$. Hence we can apply the above estimation
with $x=\bclevel{}$.

Since $\bclevel{} \leq \log(n) - \log \log (n) + O(1)$ w.h.p.\ by \cref{lem:fmv_tuple_tailbound}
we obtain that, w.h.p., \[ b(\bclevel{})=x + 3 \log (x) + 32 \leq \log(n) + 2 \log \log n + O(1).\]
So, for sufficiently large $n$,
w.h.p., \[b(\bclevel{}) \leq \log n + 2\log \log n + O(1),\]
so that \[2^{b(\bclevel{})} = n \log^2 n \cdot 2^{O(1)} = O(n \log^2 n).\]
The claim follows by union bound and multiplying the state counts of all fields together.
\end{proof}

In the following observation we bound the number of red and blue tokens coming from a correct approximation by protocol $\EstLogN$.

\begin{observation} \label{obs:number-of-tokens}
Assume \cref{alg:EstLogN} ($\EstLogN$) has stabilized with a correct approximation of $n$.
Let $R=2^{r(\bclevel{})}$ denote the number of red tokens generated by \CountingProt\ and let $B=2^{b(\bclevel{})}$ denote the number of blue tokens generated by \cref{alg:balance} (\Balance)
for each red token. Then we have, for $n$ large enough,
\begin{enumerate}\itemsep1ex
\item $\displaystyle \frac{n}{2^{14}(c\log n + 2 \log\log n)}\le R\le \frac{4^c}{128} n$ with probability $1-O(1/n^{c})$.

\item $\displaystyle \frac{2^{26} n (\log n)^2}{c\log n + 2 \log\log n} \le B\le\frac{2^{36} 4^c n (\log n + 2c+4)^3}{\log n}$ with probability $1-O(1/n^{c})$.

\item The total number of blue tokens $M$ generated by \cref{alg:balance} (\Balance) is at least
\[
M = R\cdot B \ge
\frac{2^{12} n^2 (\log^2 n)}{(c\log n + 2 \log\log n)^2} \ge
\frac{2^{12} n^2 (\log^2 n)}{(2c\log n)^2} =
\frac{2^{10}}{c^2}\cdot n^2, \ \mbox{w.h.p.\ (for $c \geq 2$)}
\]
\end{enumerate}
\end{observation}
\begin{proof}
These are tedious but otherwise straightforward calculations based on \cref{eq:b} and the detailed bounds of \cref{lem:fmv_tuple_tailbound}.
\end{proof}

\subsection{Balancing Blue Tokens}

\begin{lemma}\label{lem:blue-balance-time}
Let $T_0$ be the first time when the $\EstLogN$ has stabilized on the correct value and
let $T_2\geq T_0$ be the first time when all red tokens are converted into blue tokens (i.e., \mbox{$\logred v=\bot$}).
We define $T_3$ to be the first time when $\max_u\{\blue u\} - \min_v \{ \blue v \} \leq 3$.
Then w.h.p.\ \[ T_3 - T_2 = O\left(\log(n) \cdot \loadbalancingtime(G)\right).\]
\end{lemma}

\begin{proof}
Let $B_u(t) = \blue{u}(t)$ be the number of blue tokens of node $u$ after step $t$ and define $\vec{B}(t) = (B_u(t))_{u \in V}$.
Observe that the total number of blue tokens $\lVert \vec{B}(t) \rVert_{1}$ can only change if a red token converts into multiple blue tokens. Therefore, after time $T_1$ the total number of blue tokens remains fixed with $\lVert \vec{B}(t) \rVert_{1} = \lVert \vec{B}(T_1) \rVert_{1}$ for any $t \geq T_2$.

The \emph{discrepancy} of a load vector $\vec{B}(t)$ is defined as $\max_u\{ B_u(t) \} - \min_v\{ B_v(t) \}$.
A red token can only convert into blue tokens if $\blue u \leq 2 \cdot 2^{b(\bclevel u)} + 3$ (see Line \ref{ln:explosion-condition} in \cref{alg:counting}). Hence, $B_u(T_1) \leq 3 \cdot 2^{b(\bclevel u)} + 3 $ for any node $u$ and thus $\disc(\vec{B}(T_1)) = O(\poly n)$ w.h.p.

Our goal is to apply a variant of Theorem 2.3 from \cite{sodaPaper} for the following \emph{asynchronous} load balancing process with randomized rounding.
For a given load vector $ \vec{B}(t)$ at time $t$, the load vector $\vec{B}(t+1)$ is obtained from $\vec{B}(t)$ by choosing an edge $\{u, v\} \in E(G)$ uniformly at random and balancing the loads of nodes $u$ and $v$ over this edge as evenly as possible.
So-called \emph{excess tokens} which arise due to load being discrete are assigned randomly to either $u$ or $v$.
Formally,
\[(B_u(t+1), B_v(t+1)) = \begin{cases}
\left(\left\lceil \frac{B_u(t)+B_v(t)}{2} \right\rceil, \left\lfloor \frac{B_u(t)+B_v(t)}{2} \right\rfloor\right),\quad\text{with probability $1/2$,} \\[1ex]
\left(\left\lfloor \frac{B_u(t)+B_v(t)}{2} \right\rfloor, \left\lceil \frac{B_u(t)+B_v(t)}{2} \right\rceil\right),\quad\text{with probability $1/2$,}
\end{cases}\]
while $B_w(t+1) = B_w(t)$ for all $w \not\in \{u, v\}$.

Observe that this load balancing process is equivalent to our population protocol for balancing the blue tokens. Indeed, in our protocol we choose an \emph{oriented} edge $(u, v)$ uniformly at random
and assign the rounded-up value to the initiator $u$ and the rounded-down value to the responder $v$, with both edge orientations $(u,v)$ and $(v,u)$ chosen with the same probability $1/2$.
Observe furthermore that $\disc{\vec{B}(T_1)} = O(\poly{n})$. We apply Theorem 2.3 from \cite{sodaPaper} (and the remarks that follow the theorem) to the load balancing process with initial load vector $\vec{B}(T_1)$.
This gives us for a $\tau = O(n \log(n) \cdot \frac{d(G)}{\Delta(G)} \cdot (1-\lambda(\mP(G))))$ and any constant $c > 0$ that
\[\Pr[\disc(\vec{B}(\tau))\leq 3] \geq 1 - \exp(-\log^{1-c}(n)).\]
The lemma now follows from the definition $\loadbalancingtime(G) = n \log(n) \frac{d(G)}{\Delta(G)}(1-\lambda(\mP(G)))^{-1}$, where the high probability is achieved by repeating the load balancing process up to $\log n$ times.
\end{proof}

\newcommand{\tauepoch}{\tau_{\mathrm{epoch}}}

\subsection{Eliminating Red Tokens}

\begin{lemma}\label{lem:eliminating-red-tokens}
Let $T_0$ be the first time when the $\EstLogN$ has stabilized on the correct value and
let $T_1 \geq T_0$ be the first time in which all
agents have at most one red token (i.e., $\logred{v} \leq 0$). We define $T_2$ to be the first time at which all
red tokens are converted into blue tokens (i.e., $\logred{v} =\bot$).
Then w.h.p. \[ T_2 - T_1 = O(\LBT \cdot \log n).\]
\end{lemma}

\begin{proof}
Since at time $T_1$ each node holds at most one red token. For ease of presentation we call the nodes holding a red token \emph{red nodes}.
On a high level, the proof proceeds by dividing the execution into non-overlapping epochs and showing that in every second epoch the number of red nodes decreases by a constant factor with constant probability. Recall that a red node will be able to convert its red token into blue tokens as soon as a node interacts with a node holding a small enough number of blue tokens. Fortunately we are able to show that there is a sufficient number of nodes which can be used for the conversion since agents constantly balance their blue tokens.

\smallskip

We partition time into epochs of length~$2\cdot\LBT$.
For $i\ge 0$ let $t_i$ denote the first step of epoch~$i$, and let $R(i)$ denote the number of red nodes at time~$t_i$.
For our proof it will be necessary to distinguish between the blue
tokens a node has at the beginning of an epoch and the blue tokens it receives due to conversions in that phase.
Hence, we refer to the tokens which are on a node at the beginning of an epoch as \emph{old blue tokens}.
Respectively, all blue tokens added in epoch $i$ are new blue tokens.
In the following we refer to the set of old blue tokens which are on the nodes at the beginning of epoch $i$ as $B'(i)$.
Note that we can assume w.l.o.g.\ that our protocol balances the old blue tokens and the \emph{new} blue tokens independently.
By \cite{sodaPaper}, the old blue tokens from epoch $i$ are balanced at the beginning of epoch $i+1$ with constant probability,
and the discrepancy is at most~$3$ .
Since the total number of blue tokens generated by \cref{alg:balance} ($\Balance$) is bounded by $\BN$
there are at most $\BN$ old blue tokens. Hence, each node holds at most $\RedToBlue+3$ old blue tokens at the end of each epoch.

Now consider an arbitrary red node $u$ at time~$t_{i+1}$, i.e., the beginning of epoch $(i+1)$.
In the following we call a red node \emph{good} in step $t$ if its number of blue tokens is small enough
such that the node is able to convert its red token into blue ones during its next interaction.
Formally, node $u$ is \emph{good} in step $t$ if it holds $\RedToBlue{}$ or fewer blue tokens in step $t$.
Otherwise we call $u$ \emph{blocked} in step $t$.

Our goal is to show that $R(i+2)\le \tfrac{2}{3} R(i)$ for all even $i$.
We now distinguish three cases depending on the number of red nodes $R(i+1)$.
\begin{enumerate}
\item If $R(i+1) \le \tfrac{2}{3} R(i)$ we are done as the number of red nodes is monotonically decreasing.

\item Suppose $R(i+1) \ge \tfrac{2}{3} R(i)$ and at least half of the red nodes in $R(i+1)$ are good at time $t_{i+1}$.
We want to show that these nodes convert their token on their next interaction.
For this, they must (still) have less than $\PoofBound{}$ after their next interaction.

By the definition, good red nodes have ${\RedToBlue{}}$ or fewer blue token.
Further, any other node has at most $\PoofBound + \RedToBlue \leq 3(\RedToBlue +1)$ blue token in any step.
This can be easily verified by contradiction:
Let $v$ be the first node with more than $\PoofBound + \RedToBlue$ blue tokens.
Suppose that the node received the tokens in \Cref{ln:explosion} of \Cref{alg:balance}.
Then, it would have executed the line with more than $\PoofBound$ blue token.
This contradicts \Cref{ln:explosion-condition} of \Cref{alg:balance}.
Otherwise, if received the tokens in \Cref{ln:balanceblue}, either $v$ or its interaction partner must have had $\PoofBound + \RedToBlue$ blue tokens before.
Thus, $v$ is \emph{not} the first node with that condition.
This contradicts our assumption.
Thus, no matter with which node they interact, they are able to convert their red token into blue ones during the next interaction as
\begin{align*}
\frac{1}{2}\left(\RedToBlue + 3(\RedToBlue +1)\right) \leq \PoofBound.
\end{align*}
Thus, the node meets the condition of \Cref{ln:explosion-condition} and converts the red token.

As within \LBT\ steps each node interacts w.h.p.\ at least once, again w.h.p.\ at least $\tfrac{1}{3}R(i)$ red tokens are eliminated.

\item Suppose $R(i+1) \ge \tfrac{2}{3} R(i)$ and fewer than half of the red nodes in $R(i+1)$ are good at time $t_{i+1}$.
In the following we show that this case cannot happen.

From the definition of goodness it follows that $\ge \tfrac{1}{3}R(i)$ of the red nodes are blocked at time $t_{i+1}$ and
each of them has $\PoofBound$ blue tokens.
Since at most $\RedToBlue + 3$ of these tokens are old blue tokens at time $t_{i+1}$
and each blocked node has at least $\PoofBound$ blue tokens in step $t_{i+1}$, it holds there are
\[
\tfrac{1}{3}R(i)\cdot \left(\PoofBound - (\RedToBlue + 3) \right) = \tfrac{1}{3}R(i) \cdot \RedToBlue
\]
``new'' blue tokens in step $t_{i+1}$ due to conversions in epoch $i$.
Since each red token is converted to at most $\RedToBlue$ ``new'' blue tokens, the number of red tokens eliminated during epoch $i$
is at least
\[
\frac{\tfrac{1}{3}R(i)\cdot \RedToBlue}{\RedToBlue}
\;\ge\;
\tfrac{1}{3}R(i).
\]
This contradicts the assumption that $R(i+1) \ge \tfrac{2}{3}R(i)$.
\end{enumerate}

As regards the probability, with constant probability the number of red nodes decreases every two epochs by at least a constant factor. The result follows immediately by considering $\Theta(\log n)$ epochs and applying Chernoff bounds.
\end{proof}

\subsection{Balancing Red Tokens}\label{sec:red_token_analysis}

\begin{lemma}\label{lem:red-balance-time}
Let $T_0$ be the first time when protocol $\EstLogN$ has stabilized on the correct value
and let $T_1$ be the first time when all
agents have at most one red token (i.e., $\logred{u} \leq 0$).
Then w.h.p. \[ T_1 - T_0 = O(\log(n) \cdot \loadbalancingtime(G)).\]
\end{lemma}

First we present an overview of our proof.
The main idea of the proof is to track the positions of individual red tokens and to bound the time until each token is the sole red token on its host node.

We number the red tokens as integers $1$ through $k = \sum_{v} 2^{\logred{v}(T_0)}$.
Recall that, if a node $u$ has $2^{\logred{u}} \ge 2$ tokens and its interaction partner $v$ none at all,
the tokens are evenly split between $u$ and $v$ (see lines \ref{alg:balance:2} to \ref{alg:balance:5} in \cref{alg:balance} $\Balance$).
We call such an operation a \emph{halving step},
and model the locations by moving a uniformly random subset of $u$'s tokens of size $2^{\logred{u} - 1}$ from $u$ to $v$.
If, in line \ref{alg:balance:7} of $\Balance$,
$u$ and $v$ swap their $\logred{}$ values,
we say that all tokens exchange places, except if $\logred{u} = \logred{v}$, then we only do so with probability $1/2$.

Under this assignment of locations,
every red token, when considered by itself, performs a lazy random walk on the graph until it is converted
(see \cref{obs:random_walk}).
Note that initially all red tokens are on the same node,
and that tokens which are on the same node travel together until they are separated due to a halving step.
However, once tokens are on separate nodes,
they can never be on the same node again.

To prove the theorem we fix a red token $i$ and consider disjoint epochs of length $O(\loadbalancingtime(G))$.
Now we consider, for each epoch $j$, the last interaction of token $i$ in epoch $j$. Our goal is show that there is a constant probability of $i$ being involved in a halving step in exactly that interaction (unless, of course, it has already halved earlier in epoch $j$).
Since in the beginning, if $\EstLogN{}$ stabilized on a correct value,
there are at most $n/c$ red tokens (for a constant $c$)
$i$ must be alone after $\log_2 n$ halving steps.

It remains to show that for token $i$ the probability of a halving step is constant per epoch.
A halving step occurs exactly when the node holding token $i$
interacts with a node not having any red tokens.
Since there is only a constant fraction of $n/c$ red tokens (see \cref{obs:number-of-tokens}) there must be many such nodes without red tokens.
If halving steps only occurred at the end of an epoch, the stacks of red tokens would be distributed (more or less) uniformly among the nodes of the graph by the end of the epoch (see \cref{lem:tokens_negative_dependence}). Hence, it would be straightforward to show that interacting with an empty node would happen with constant probability, and we would be done.
The usual method to deal with such situations
is to simply assume that, since there are at most $n/c$ red tokens, the worst-case probability to interact during the last step of an epoch with an empty node is still $1-1/c$. However, stacks of red tokens can be halved at any time during an epoch, and we cannot focus on the last step due to significant dependencies between red tokens which started out on the same node at the beginning of the epoch.

\smallskip

We address the problem as follows (see \cref{lem:bound_halving}).
Let $I$ be the set of red tokens which are, at the start of epoch $j$, in the same stack of red tokens (at the same node) as token $i$.
Tokens in $I$ stay together until they undergo a halving step. Tokens in $I$ and tokens not in $I$ do not move completely independently at random, since they can never be be at the same node, but they move independently otherwise (we make this precise in \cref{lem:tokens_negative_dependence}).
Now let $v$ be the node
interacting with the node holding $i$ in the final interaction of the epoch.
There are three disjoint cases:
$v$ has no tokens at all, $v$ has tokens from $I$,
or $v$ has tokens, but not from $I$.
The first case is easy:
there will be a halving step then and there (unless, of course, $i$ is already alone, in which case there must have been a halving step before).
In the second case it is easy to see that there has been a halving step before,
since the tokens in $I$ are no longer all on the same node.
Hence, only the third case is problematic.
In this case we bound the probability of a halving step by using the fact that tokens starting at different nodes behave almost independently (again \cref{lem:tokens_negative_dependence})
This allows us to bound the expected number of tokens at $v$ not from $I$ by a constant less than $1$.

One complication we have omitted so far is that there are also blue tokens on the nodes and red tokens can only be converted into blue tokens if there are not too many of them at the node (lines \ref{ln:explosion-condition} and \ref{ln:explosion} in \cref{alg:balance} $\Balance$). This is somewhat annoying to deal with formally, but otherwise unproblematic. Apart from this, the red tokens behave independently of blue tokens, so that it is fine to consider the red tokens by themselves.

\newcommand{\AloneTime}[1]{T^{\mathrm{alone}}_{#1}}

\begin{proof}[Proof of \cref{lem:red-balance-time}]
For each token $i$,
let $\AloneTime{i} \geq T_0$ be the first time where $i$ is alone (or converted into blue tokens, if those happen in the same step).
As we have $\logred{v}(T_0 + \tau) \leq 0$ for all $v$
exactly when all tokens are either alone or have been converted to blue tokens,
and tokens must be alone before being converted,
we have $T_1 = \max_{i \in [k]} \AloneTime{i}$.
To bound $\AloneTime{i}$ for a fixed token $i$,
we divide time into disjoint epochs as follows:
Let $T_{0,i} = T_0$,
and let, for $j \geq 1$,
\[T_{j,i}=\begin{cases}T_{j,i-1} + \loadbalancingtime(G),\quad&\textup{if $\TokenPos{i}{T_{j,i-1}+\loadbalancingtime(G)} = \bot$},\\
\min\{t \geq T_{j,i+1} + \loadbalancingtime(G) \mid \TokenPos{i}{T_{j,i+1} + \loadbalancingtime(G)} \in e(t)\},\quad&\textup{otherwise.} \end{cases}\]
Then epoch $j$ is the time interval $[T_{j-1,i}+1,T_{j,i}]$.
Call epoch $j \geq 1$
successful
if token $i$ is alone at the beginning of the epoch,
if it has been converted into blue tokens by the end of the epoch,
or if it has undergone a halving step during the epoch.

Let $E_i$ be the index of the $\lfloor \log n \rfloor$th successful epoch for $i$.
Now at time $T_{E_i,i}$,
token $i$ is either alone or has been converted into blue tokens.
This is because
$i$ becomes alone exactly when it participates in its $\logred{v}(T_0)$th halving step:
There are $2^{\logred{v}(T_0)} \leq n$ tokens at $i$'s position at time $T_0$ (including $i$),
and in general $2^{\logred{v}(T_0+\tau_0)-j}$ tokens at $i$'s position after the $j$th halving step.
Hence $\AloneTime{i} \leq T_{E_i,i}$.

To bound the latter,
by \cref{lem:bound_halving,lem:bound_concrete_probs} (which we state and prove below),
the probability that an epoch is successful
is at least $\frac{1}{40} - n^{-6} = \Omega(1)$.
And so we have $E_i = O(\log n)$ w.h.p.
And $T_{E_i,i} - T_{E_i,i} - T_0$ is $E_i \cdot \loadbalancingtime(G)$ plus the time until a node interacts $E_i$ times.
Since $E_i \geq \lfloor \log n \rfloor$,
the time until a node interacts $E_i$ times is in $O(n E_i)$ w.h.p.,
which is in $O(n \log(n)) \subseteq O(\loadbalancingtime(G))$ w.h.p.
And so $T_{E_i,i} - T_{0,i} = O(\log(n) \loadbalancingtime(G))$ with probability at least $1 - n^{-2}$.
The claim then follows by union bound over all $i$.
\end{proof}

So it remains to state and prove the lemmas referred to in the preceding proof.
Throughout, we write $\TokenPos{i}{\tau}$ for the position of token $i$ at time $T_0 + \tau$,
or $\bot$ if token $i$ has been converted into blue tokens at or before time $T_0 + \tau$.
The following lemma gives two ways to bound from below the probability of a token $i$ being involved in a halving step in some time interval.
The idea is to observe the first interaction of $i$'s location after some fixed time.
Then the probability of a halving step occurring at or that interaction can be bounded from below both by the probability of the node holding $i$ interacting with a node having a token
which started at a different location from $i$,
or alternatively by the probability of the token having been converted into blue tokens already.

\begin{lemma}\label{lem:bound_halving}
Fix a sequence of (undirected) edges $(e(\tau))_{\tau=\tau_0}^\infty$
chosen by the scheduler (but not their orientations).
Let $\vec{w}$ be a location vector,
and let token $i$ be a token that is not alone in $\vec{w}$.
Let $T > \tau_0 + \loadbalancingtime(G)$
be the first step in which token $i$ participates in an interaction, and let $U$ be the random variable denoting the interaction partner (or $\bot$ if no such interaction exists).

Define the following events:\begin{itemize}
\item $\Ev{H}$: token $i$ participates in a halving step in the time interval $[\tau_0, T]$
\item $\Ev{O}$: there is a token $j$ with $w_j \neq w_i$ and $\TokenPos{j}{T-1} = U \neq \bot$
\end{itemize}
Then
\begin{equation}\label{eq:bound_halving}
\Pr[\Ev H \mid \vTokenPos{\tau_0} = \vec{w}] \geq \max\{1 - \Pr[\Ev{O} \mid \vTokenPos{\tau_0} = \vec{w}], \Pr[\TokenPos{i}{\tau_0 + \loadbalancingtime(G)} = \bot \mid \vTokenPos{\tau_0} = \vec{w}]\}.\end{equation}
\end{lemma}

\begin{proof}
We condition on $\vTokenPos{\tau_0} = \vec{w}$ and the edge sequence throughout.
We write $I = \{j \mid w_j = w_i\}$
for the set of tokens sharing a location with $i$ in $\vec{w}$ (including $i$ itself),
and $O = \{j \mid w_j \neq w_i\}$ for the set of tokens \emph{not} sharing a location with $i$ in $\vec{w}$.

We define event $\Ev{I}$
as the event that there is a token $j$ with $w_j = w_i$ and $\TokenPos{j}{T_1} = U \neq \bot$.
Furthermore we let $\Ev{empty}$ be the event that $U$ is empty at time $T-1$ (i.e., immediately before the interaction at time $T$ happens).
Lastly, we let $\Ev{\bot}$ be the event that $\TokenPos{i}{\tau_0 + \loadbalancingtime(G)} = \bot$.

As tokens from $I$ and tokens from $O$ start at different positions at time $\tau_0 \leq T$,
no node $v$ will ever have a token in $I$ and a token in $O$ at time $T$.
So $\Ev{I}$ and $\Ev{O}$ are mutually exclusive,
as are all pairings with $\Ev{\bot}$ and $\Ev{empty}$ by definition of the events.
And the only possibility not covered by these four events
is that token $i$ is placed at some node,
but that node never interacts anymore.
But this occurs only with probability $0$ for the random scheduler, so we can safely ignore this case.
So we have \begin{equation}\label{eq:not_quite_partition}\Pr[\Ev{\bot}] + \Pr[\Ev{I}] + \Pr[\Ev{O}] + \Pr[\Ev{empty}] = 1.\end{equation}

Now first, as $i$ was not alone at time $\tau_0$,
we must have $\Ev{\bot} \subseteq \Ev{H}$.
This is because
a token can only be deleted at time $\tau_0 + \loadbalancingtime(G)$
if it was alone at a prior time,
implying that it must have participated in a halving step in the time interval $[\tau_0, \tau_0 + \loadbalancingtime(G)]$.

Next, consider $\Ev{empty}$ holding.
If $i$ is alone at time $T$,
then it must have participated in a halving step since time $\tau_0$ (as it was not alone then),
so that $\Ev H$ holds.
If $i$ is not alone at time $T$,
then it participates in a halving step at time $T$ (since $U$ is empty), so that $\Ev H$ also holds.
So \(\Ev{empty} \subseteq \Ev H.\)

Lastly, if $\Ev I$ holds,
then at time $T$ there are tokens in $I$ which are on a different node than $i \in I$.
This necessarily means that $i$ had participated in a halving step between time $\tau_0$ and $T$ (since all tokens in $I$ were at the same node at time $\tau_0$),
so that \(\Ev I \subseteq \Ev H.\)

This all together implies
\(\Ev \bot \cup \Ev I \cup \Ev{empty} \subseteq \Ev H,\)
and so combined with \cref{eq:not_quite_partition} and disjointness of $\Ev \bot$, $\Ev I$, and $\Ev{empty}$ we get
\[\Pr[\Ev H] \geq \Pr[\Ev{\bot}] + \Pr[\Ev{I}] + \Pr[\Ev{empty}] = 1 - \Pr[\Ev{O}]\]
which directly implies the claim.
\end{proof}

In order to derive bounds for the right-hand-side of \cref{eq:bound_halving},
we consider the random walks
which are performed by not-yet-converted tokens:
When an edge at the node where a token is placed is chosen by the scheduler,
the token moves to the other side with probability $1/2$,
and stays in place otherwise (and perhaps is converted afterwards).
More formally,
fix the sequence $(e(\tau))_{\tau}$ of edges chosen by the scheduler (but not their orientations!),
and let $\mathbf{M}^{(\tau)} \in [0,1]^{V \times V}$ be the stochastic matrix given by
\[\mathbf{M}_{u,v}^{(\tau)} = \begin{cases}
\frac{1}{2},\quad\textup{if $\{u, v\} \cap e(\tau) \neq \emptyset$,} \\
1,\quad\textup{if $u = v \not\in e(\tau)$,} \\
0,\quad\textup{otherwise.}
\end{cases}\]
Furthermore, let
$\mathbf{M}^{[\tau_1, \tau_2]} = \mathbf{M}^{(\tau_1)} \cdot \mathbf{M}^{(\tau_1 + 1)} \cdots \mathbf{M}^{(\tau_2)}$.

\begin{observation}\label{obs:random_walk}
Fix an edge sequence $(e(\tau))_{\tau=\tau_0}^\infty$, fix a token $j$ and two nodes $u, v$.
Then for any $\tau$,
\[\Pr[\TokenPos{j}{\tau} = v \mid \TokenPos{j}{\tau-1} = u] \leq \mathbf{M}_{u,v},\]
and more generally for any $\tau_1 < \tau_2$
\[\Pr[\TokenPos{j}{\tau_2} = v \mid \TokenPos{j}{\tau_1-1} = u] \leq \mMix{u}{v}{\tau_1}{\tau_2}.\]
\end{observation}

\noindent Note that the inequalities in \cref{obs:random_walk} would hold with equality, if tokens could not be deleted.
Of course, these random walks of different tokens are not independent.

The next lemma now bounds the probabilities on the right-hand-side of \cref{eq:bound_halving}.

\begin{lemma}\label{lem:bound_concrete_probs}
In the setting of \cref{lem:bound_halving}, we have
\[\max\{1 - \Pr[\Ev{O} \mid \vTokenPos{\tau_0} = \vec{w}], \Pr[\TokenPos{i}{\tau_1} = \bot \mid \vTokenPos{\tau_0} = \vec{w}]\} \geq \frac{1}{40} - \frac{1}{n^{-6}}.\]
\end{lemma}

\begin{proof}
Let $u = w_i$ be the initial position of token $i$,
and note that $\Pr[\TokenPos{i}{t_1} = v \mid \vTokenPos{t_0} = \vec{w}] = \Pr[\TokenPos{i}{t_1} = v \mid \TokenPos{i}{t_0} = u] \leq \matMix{t_0}{t_1}$ for any $t_1 > t_0$ (see \cref{obs:random_walk}).

First,
note that by the second display equation in the proof of Corollary 8.3 in the full version of \cite{sodaPaper} and definition of $\loadbalancingtime(G)$,
for any $t_0$,
we have \[\Pr\left[\max_{u} \left\|\matMix{t_0+1}{t_0+\loadbalancingtime(G)}_{u,\cdot} - \frac{\vec{1}}{n}\right\|_2^2 \leq \frac{1}{n^7}\right] \geq 1 - n^{-6},\]
and consequently $\Pr[\max_{u,v} |\matMix{t_0+1}{t_0 + \loadbalancingtime(G)}_{u,v} - \frac{1}{n}| \leq n^{-3.5}] \geq 1-n^{-6}$.
As $n \geq 2$,
we have $n^{-3.5} \leq 2^{-2.5}/n < 5/n$,
so that
\begin{equation}\label{eq:mix_bound}\Pr\left[\max_{u,v} \left|\matMix{t_0+1}{t_0 + \loadbalancingtime(G)}_{u,v} - \frac{1}{n}\right| \leq \frac{1}{5n}\right] \geq 1 - n^{-6}.\end{equation}

Now assume that there more than $n/2$ nodes $v$ for which $\Pr[\TokenPos{i}{t_1} = v \mid \TokenPos{i}{t_0} = u] < 0.75/n$, and let $V'$ be the set of such nodes.
Then
\begin{align*}
\Pr[\TokenPos{i}{t_1} = \bot \mid \TokenPos{i}{t_0} = u]
&= 1 - \sum_{v\in V} \Pr[\TokenPos{i}{t_1} = v \mid \TokenPos{i}{t_0} = u]
\\ &= \sum_{v\in V} \mathbf{M}_{u,v} - \sum_{v\in V} \Pr[\TokenPos{i}{t_1} = v \mid \TokenPos{i}{t_0} = u]
\\ &\geq \sum_{v\in V'} (\mathbf{M}_{u,v} - \Pr[\TokenPos{i}{t_1} = v \mid \TokenPos{i}{t_0} = u])
\\ &> \sum_{v\in V'} (\mathbf{M}_{u,v} - (\mathbf{M}_{u,v} - 0.05/n)) - n^{-6}
> \frac{n}{2} \cdot \frac{0.05}{n} - n^{-6}
\\ &= \frac{1}{40} - n^{-6},
\end{align*}
where the first strict inequality uses \cref{eq:mix_bound}.

Otherwise, there are at most $n/2$ nodes $v$ where $\Pr[\TokenPos{i}{t_1} = v \mid \TokenPos{i}{t_0} = u] \geq 0.75/n$;
let $V'$ be the set of such $v$.
Then
\[\Pr[\TokenPos{i}{t_1} \in V' \mid \TokenPos{i}{t_0} = u] \geq \frac{n}{2} \cdot \frac{0.75}{n} = \frac{3}{8}.\]
And for all $u \in V'$, all $v \neq u$, and $j \in O$:
\begin{align*}
\Pr[\TokenPos{j}{t_1} = v \mid \TokenPos{i}{t_1}=u]
&= \frac{\Pr[\TokenPos{i}{t_1}=u \wedge \TokenPos{j}{t_1} = v]}{\Pr[\TokenPos{i}{t_1}=u ]}
\\ &\leq \frac{\Pr[\{\TokenPos{i}{t_1}, \TokenPos{j}{t_1}\}\subseteq \{u, v\} ]}{0.75/n}
\\ &\leq \frac{4n}{3} \cdot \Pr[\TokenPos{i}{t_1} \in \{u, v\} ] \cdot \Pr[\TokenPos{j}{t_1} \in \{u, v\} ]
\\ &\leq \frac{4n}{3} \cdot \frac{2.4}{n} \cdot \frac{2.4}{n}
< \frac{8}{n},
\end{align*}
where the first inequality is by implication and the fact that $u \in V'$,
and the second inequality is by \cref{lem:tokens_negative_dependence} (using that $j \in O$), which we state and prove further below due to its technical nature.
And so in particular
\[\max_{v \neq u} \Pr[\TokenPos{j}{t_1} = v \mid \TokenPos{i}{t_1}=u] < \frac{8}{n}.\]
Now if we continue until the time $T$ of the next interaction involving node $u$,
the position of token $j$ is a Markov chain on $V \setminus \{u\} \cup\{\bot\}$,
and so \(\max_{v \neq u} \Pr[\TokenPos{j}{t+\tau} = v \mid \TokenPos{i}{t_1}=u]\)
is non-increasing in $\tau$ for $\tau < T$.
And hence, for all $v \neq u$,
\[\Pr[\TokenPos{j}{T-1} = v \mid \TokenPos{i}{t_1}=u] < \frac{8}{n}.\]
So by Markov's inequality and linearity of expectation,
as there are at most $n/16$ tokens by assumption,
\[\Pr[\{j \in O \mid \TokenPos{j}{T-1}=v \} \geq 1 \mid \TokenPos{i}{t_1}=u] < \frac{1}{2}.\]
By the law of total probability over all possible values of $U$, and since $U \neq u$
we then get
\[\Pr[\{j \in O \mid \TokenPos{j}{T-1}=U \} \geq 1 \mid \TokenPos{i}{t_1}=u] < \frac{1}{2}.\]
So then combining the above we get,
\begin{align*}
\MoveEqLeft \Pr[\{j \in O \mid \TokenPos{j}{T-1}=U\} \geq 1]
\\&= \Pr[\{j \in O \mid \TokenPos{j}{T-1}=U\} \geq 1 \mid \TokenPos{i}{t_1} \in V'] \cdot \Pr[\TokenPos{i}{t_1} \in V']
\\ &+ \Pr[\{j \in O \mid \TokenPos{j}{T-1}=U\} \geq 1 \mid \TokenPos{i}{t_1} \not\in V'] \cdot \Pr[\TokenPos{i}{t_1} \not\in V']
\\ &< \frac{1}{2} \cdot \Pr[\TokenPos{i}{t_1} \in V']
+ 1 \cdot \Pr[\TokenPos{i}{t_1} \not\in V']
\\ &= \frac{1}{2} + \frac{1}{2} \cdot (1 - \Pr[\TokenPos{i}{t_1} \in V'])
\\ &\leq \frac{1}{2} + \frac{1}{2} \cdot \frac{5}{8}
= \frac{13}{16}.
\end{align*}
This yields the claim as $\max\{1 - \frac{13}{16}, \frac{1}{40} - n^{-6}\} = \frac{1}{40} - n^{-6}.$
\end{proof}

The following and final lemma for this section, which we deferred due to its technical nature, is an adaptation of Lemma 5.1 from the full version of \cite{sodaPaper}.
\begin{lemma}\label{lem:tokens_negative_dependence}
Fix the sequence of edges chosen by the scheduler (but not their orientations).
Let $i \neq j$ be tokens,
let $\vec{w}$ be a location vector where $w_i \neq w_j$.
Then for all $u \neq v$ and $t > t_0$:
\begin{align*}\MoveEqLeft\Pr[\{\TokenPos{i}{t}, \TokenPos{j}{t}\}\subseteq \{u, v\} \mid \vTokenPos{t_0} = \vec{w}]
\\& \leq \Pr[\TokenPos{i}{t} \in \{u, v\} \mid \vTokenPos{t_0} = \vec{w}]
\cdot \Pr[\TokenPos{j}{t} \in \{u, v\} \mid \vTokenPos{t_0} = \vec{w}].\end{align*}
\end{lemma}

\begin{proof}
We condition on $\vTokenPos{t_0} = \vec{w}$ throughout and omit writing out this conditioning.

Let $S = \{u,v\}$. 
To reduce clutter,
we write, for $b > a$
\[\iprob{a}{b} = \Pr[\TokenPos{i}{t} \in S \mid \TokenPos{i}{b-1}=\TokenPos{i}{a}]\]
and
\[\jprob{a}{b} = \Pr[\TokenPos{j}{t} \in S \mid \TokenPos{j}{b-1}=\TokenPos{j}{a}]\]
Let $\filter{\tau}$ be the filtration
revealing the random choices up to time $\tau$,
and hence in particular the locations $\vTokenPos{\tau}$ at time $\tau$.
We show below for all $\tau \in [t_0, t]$ that
\begin{align}\label{eq:Z_supermartingale}
\E{\iprob{\tau}{\tau+1} \cdot \jprob{\tau}{\tau+1} \mid \filter{\tau-1}}
\leq \iprob{\tau-1}{\tau} \cdot \jprob{\tau-1}{\tau}.
\end{align}
Now $\iprob{t}{t+1} = \1_{\TokenPos{i}{t} \in S}$ and $\jprob{t}{t+1} = \1_{\TokenPos{j}{t} \in S}$ since here,
the conditioning is trivial.
So the claim follows by inductively applying \cref{eq:Z_supermartingale}:
\begin{align*}
\Pr[\{\TokenPos{i}{t}, \TokenPos{j}{t}\}\subseteq \{u, v\}]
&= \E{\iprob{t}{t+1} \cdot \jprob{t}{t+1}}
\leq \iprob{t_0}{t_0+1} \cdot \jprob{t_0}{t_0+1}\\
&= \Pr[\TokenPos{i}{t} \in S] \cdot \Pr[\TokenPos{j}{t} \in S],
\end{align*}
where the last equality crucially uses the implicit conditioning on $\vec{W}^{t_0} = \vec{w}$.

So it remains to show \cref{eq:Z_supermartingale}.
Note that $\E{\iprob{\tau}{\tau+1} \mid \filter{\tau-1}} =\iprob{\tau-1}{\tau}$
and $\E{\jprob{\tau}{\tau+1} \mid \filter{\tau-1}} =\jprob{\tau-1}{\tau}$ by definition,
and that $\E{\iprob{\tau-1}{a} \mid \filter{\tau-1}} =\iprob{\tau-1}{a}$ and $\E{\jprob{\tau-1}{a} \mid \filter{\tau-1}} =\jprob{\tau-1}{a}$
for any $a$ since the locations of $i$ and $j$ at time $\tau-1$ are determined under $\filter{\tau-1}$
so that $\iprob{\tau-1}{a}$ and $\jprob{\tau-1}{a}$ are deterministic.
We write $e(\tau)$ for the edge chosen by the scheduler at time $\tau$.

If $\TokenPos{j}{\tau-1} \not\in e(\tau)$, then
$j$ cannot move nor be deleted at time $\tau$,
so we have $\jprob{\tau-1}{\tau} = \jprob{\tau}{\tau+1}$, and hence
\begin{align*}
\E{\iprob{\tau}{\tau+1} \cdot \jprob{\tau}{\tau+1} \mid \filter{\tau-1}}
&= \E{\iprob{\tau}{\tau+1} \cdot\jprob{\tau-1}{\tau} \mid \filter{\tau-1}}\\
&= \E{\iprob{\tau}{\tau+1} \mid \filter{\tau - 1}} \cdot\jprob{\tau-1}{\tau}
=\iprob{\tau-1}{\tau} \cdot\jprob{\tau-1}{\tau}.
\end{align*}
And by symmetry, the same holds when $\TokenPos{i}{\tau-1} \not\in e(t)$.

Otherwise we have $\{\TokenPos{i}{\tau-1}, \TokenPos{j}{\tau-1}\} = e(\tau)$.
In that case, we have
\[(\TokenPos{i}{\tau}, \TokenPos{j}{\tau}) = \begin{cases}
(\TokenPos{j}{\tau-1}, \TokenPos{i}{\tau-1}) \quad\text{with probability at most $1/2$,}\\
(\TokenPos{i}{\tau-1}, \TokenPos{j}{\tau-1}) \quad\text{with probability at most $1/2$,}\\
\end{cases}\]
with the remaining probability taken up by cases where at least one of the tokens $i$ and $j$ is deleted, so that they remain deleted until time $t$ and thus the associated random variables are $0$.
And so
\begin{align*}
\MoveEqLeft\E{\iprob{\tau}{\tau+1} \cdot \jprob{\tau}{\tau+1} \mid \filter{\tau-1}}
\\&\leq \frac{1}{2} \cdot \E{\jprob{\tau-1}{\tau+1} \cdot \iprob{\tau-1}{\tau+1} \mid \filter{\tau-1}}
+ \frac{1}{2} \cdot \E{\iprob{\tau-1}{\tau+1} \cdot \jprob{\tau-1}{\tau+1} \mid \filter{\tau-1}}
\\ &= \E{\iprob{\tau-1}{\tau+1} \cdot \jprob{\tau-1}{\tau+1} \mid \filter{\tau-1}}
\\ &= \iprob{\tau-1}{\tau+1} \cdot \jprob{\tau-1}{\tau+1}
\\ &\leq \left(\frac{\iprob{\tau-1}{\tau+1} + \jprob{\tau-1}{\tau+1}}{2}\right)^2
= \iprob{\tau-1}{\tau} \cdot \jprob{\tau-1}{\tau},
\end{align*}
where the final equality is because $e(\tau) = \{\TokenPos{i}{\tau-1}, \TokenPos{j}{\tau-1}\}$.
\end{proof}

\section{Impossibility of Termination}\label{sect:impossibility}

\ImpossibilityIntuition

\restateImpossibleThm*

We actually show something more general:
we are not just concerned with protocols computing the exact size of an interaction graph,
but even broader properties of the size (like parity, computing a constant-factor approximation of the size of the graph, etc.).

Formally, let $\mathcal{S}$ be a family of ``size properties'', i.e., a family of subsets of $\N$.
We say that a protocol computes $\mathcal{S}$
if its output function is of the form\footnote{Of course one may consider protocols where not all states have a corresponding output. But since we are concerned only with the output of agents which have terminated, and surely a missing output from a terminated agent would be incorrect, we assume each state having a corresponding output for simplicity.} $\rho \colon Q \to \mathcal{S}$,
where $Q$ is the state space of the protocol.
A configuration of the protocol is called \emph{terminated correctly} if
there exists a terminated agent,
and if for all agents having $\terminated{v} = \ttrue$ we have $|V| \in \rho(q_v)$ (where $q_v$ is the state of $v$).
Examples for families of size properties are $\mathcal{S}_{\mathrm{count}} = \{\{n\} \mid n \in \N\}$ or $\mathrm{S}_{\mathrm{parity}} = \{\{2k \mid k \in \N_0\}, \{2k + 1 \mid k \in \N_0\}\}$.
A protocol computing $\mathcal{S}_{\mathrm{count}}$ is equivalent to one computing the exact number of nodes,
and a protocol computing $\mathcal{S}_{\mathrm{parity}}$ is equivalent to one determining parity.

We say that a family of size properties is \emph{complete} if every $n \in \N$ is contained in some set in said family.
Note that for correct termination to be possible for all graphs, $\mathcal{S}$ must be complete.
We say a family $\mathcal{S}$ of size properties is \emph{strongly increasing} if for all $n$ there are infinitely many $n' > n$
so that all $S \in \mathcal{S}$ with $n \in S$ do not contain $n'$.
The sets in $\mathcal{S}$ may be overlapping:
For example, consider the strongly increasing family consisting of the intervals $[\max\{1, 2^{k-10}\}, 2^{k+10}]$ for $k \in \N$.
In this example, we may consider the output as an approximation of $\log(|V|)$ up to $10$.
In a correctly configuration with a terminated node, there thus may be different outputs present as long as all are correct.
The families $\mathcal{S}_{\mathrm{count}}$ and $\mathcal{S}_{\mathrm{parity}}$ we gave above are also strongly increasing families.

We say that $P$ is given (only) an arbitrary initial leader
when its initial set of configurations on a graph $G$
is of the form $I_G = \{\vec{i}^{(v)} \mid v \in V\} \subseteq \{q_0, q_1\}^V$,
where $\vec{i}^{(v)}_v = q_1$, and $\vec{i}^{(v)}_u = q_0$ for all $u \neq v$.
So $I_G$ consists of the set of all configurations where exactly one agent (the leader) is in $q_1$ and all other agents are in $q_0$.

The following lemmas imply the two parts of \cref{thm:impossible} by taking $\mathcal{S} = \mathcal{S}_{\mathrm{count}}$.

\begin{lemma}\label{prop:termination_impossible_with_leader}
Let $P$ be a graph-uniform population protocol given (only) an arbitrary initial leader.
Assume $P$ computes a strongly increasing complete family of size properties $\mathcal{S}$.

Assume there is a graph $H$ and some initial configuration such that $P$ terminates correctly with positive probability in bounded time.
Then there is an infinite sequence of graphs $(G_i)_{i \in \N}$ with $|V(G_i)| < |V(G_{i+1})|$
so that the probability of $P$ terminating incorrectly on $G_i$ is in $\Omega(1)$.
\end{lemma}
\begin{proof}
By assumption $P$ terminates on $H$ in at most $t_H$ interactions, for some value $t_H$, with some positive probability $p_H > 0$
on some initial configuration $\vec{i}_H$ in which a leader is given.
Note that $|E(H)|$, $t_H$, and $p_H$ are all constants.
Let $|V(H)| < n_1, n_2, \ldots$ be an infinite sequence such that no $S \in \mathcal{S}$ contains both $n$ and $n_i$ for any $i$.
Such a sequence exists since $\mathcal{S}$ is strongly increasing.

Now let $G_i \neq H$ be a (connected) graph of size $n_i$ containing $H$ as an induced subgraph
(i.e., there is a set $V_H \subseteq V(G_i)$ such that $G[V_H] \simeq H$),\footnote{To spell it out, there is an injective function $\phi \colon V(H) \to V(G)$ such that $\{\phi(u), \phi(v)\} \in E(G)$ if and only if $\{u, v\} \in E(H)$.}
such that the cut between $H$ and $V(G) \setminus H$ contains just a single edge.
Let $\vec{i}_{G_i}$ be the initial configuration which is equal to $\vec{i}_H$ when restricted to $H$ (so in particular the leader is in the copy of $H$ contained in $G_i$).

Now consider the subsequence of the first $t_H$ interactions of the population protocol on $G_i$
involving at least one node in $V_H$,
without regard to who is initiator or responder.
This subsequence does not contain a node in $V(G_i) \setminus V_H$
with probability
\(
\left(
1 -
(
\abs{E(H)} + 1
)^{-1}
\right)^{t_H},
\)
which is a positive constant independent of $n_i$
as $|E(H)|$ and $t_H$ are constants independent of $n_i$.
Conditioned on this event, from the perspective of the nodes in $V_H$,
the execution of the protocol on $G_i$ behaves identically to an execution of the protocol on $H$,
so that $v$ will have terminated with an output $\rho(q_v) \ni |V(H)|$ with probability $\geq p_H$.
So the overall probability of this termination is in $\Omega(1)$.
But since $G_i$ is of size $n_i$, we must have $|V(G_i)| = n_i \not\in \rho(v_q)$,
so the termination was incorrect.
\end{proof}

\begin{lemma}\label{prop:termination_impossible_without_leader}
Let $P$ be a graph-uniform population protocol with all agents initialized to a fixed initial state $q_0$
computing a strongly increasing complete family of size properties $\mathcal{S}$.

Assume there there is a graph $H$ on which $P$ terminates correctly with positive probability in bounded time.
Then there is an infinite sequence of graphs $(G_i)_{i \in \N}$ with $|V(G_i)| < |V(G_{i+1})|$
so that the probability of $P$ terminating correctly on $G_i$ is $\exp(-\Theta(|V(G_i)|))$.
\end{lemma}

\begin{proof}
We let $H$, $t_H$, $p_H$, $\vec{i}_H$, and the $n_1, n_2, \ldots$ be defined as in the proof of \cref{prop:termination_impossible_with_leader}.
We now define $G_i$ as follows:
Let $k_i = \lfloor (n_i - 1) / |V(H)| \rfloor$.
Then to a clique of $n_i - k_i \cdot |V(H)|$ nodes,
connect $k_i$ disjoint copies $H_1, \ldots, H_{k_i}$ of $H$ by a single edge each (in an arbitrary way) so that indeed $G_i$ has $n_i$ nodes.
We consider an initial configuration $\vec{i}_{G_i}$ on $G_i$
so that it is equal to $\vec{i}_H$ when restricted to any of the copies $H_1, \ldots, H_{k_i}$.

The $k$ subsequences of interactions involving at least one node in the copies $H_1, \ldots, H_k$ are disjoint, so the events whether the prefixes of length $t_H$ contain the edge to the central clique are independent.
The probability of this edge not being present in a single given subsequence is
\(\left(1 - \left(|E(H)| + 1\right)^{-1}\right)^{t_H},\)
which is a constant independent of $n_i$ since $|E(H)|$ and $t_H$ are constants.
As in the proof of \cref{prop:termination_impossible_with_leader},
conditioned on this event, from the perspective of the nodes in the respective copy of $H$,
the execution of the protocol on $G_i$ behaves identically to an execution of the protocol on $H$,
so that $v$ will have terminated with an output $\rho(q_v) \ni |V(H)|$ with probability $\geq p_H$.
So for a fixed copy of $H$, the overall probability of such a termination is bounded from below by a positive constant $c_H$.
But since $G_i$ is of size $n_i$, we must have $|V(G_i)| = n_i \not\in \rho(v_q)$,
so that such a termination is incorrect in $G_i$.

So for the protocol to terminate correctly on $G_i$,
it must be the case that \emph{none} of the $k_i$ copies of $H$ in $G_i$ have such a termination,
which happens with probability at most
\(\left(1 - c_H\right)^{k_i}.\)
Now $k_i = \Theta(n_i/|V(H)|) = \Theta(n_i)$ as $|V(H)|$ is independent of $n_i$.
that this probability is in $\exp(\Theta(n_i))$, as claimed.
\end{proof}

\section{Generating Independent Random Bits}
\label{sec:randombits}

In this section we present the full analysis of our protocol for generating random bits. For sake of readability we restate the definitions from \cref{sec:random-bits-overview} before we give the proof.

\medskip

\descriptionGeneratingRandomBits

\enlargethispage{-5\baselineskip}

{
\def\dcmGobble#1#2{}
\begin{lstlisting}[title={\Cref{alg:random-bits}.\hspace{.5em}\GenerateRandomBit($u$, $v$)}]
if $(\Color u,\Color v) = (\black,\white)$ then
    $(\Rand u, \Rand v) \gets (1, \bot)$ /*The black node $u$ is initiator, assign bit $1$ */
else if $(\Color u,\Color v) = (\white,\black)$ then
    $(\Rand u, \Rand v) \gets (\bot, 0)$  /*The black node $v$ is responder, assign bit $0$ */
else 
    $(\Rand u, \Rand v) \gets (\bot, \bot)$
$ ( \Color u, \Color v ) \gets (\black, \white) \label{alg:random-bits-ln-7}$
\end{lstlisting}
}

We now prove the following theorem, stating that the protocol produces unbiased, independent random bits and within a moderate time each node receives a large number of random bits.

\restateBitSampling*

\noindent The remainder of this section contains the proof of \cref{thm:bitsampling} and various necessary technical lemmas, claims and observations.

Before we present the formal proofs, we make the underlying probability
space explicit.
The scheduler determines a sequence of \emph{directed edges}
$\Vec{M}(t) \in \bigl\{ (a,b) \mid \{a,b\} \in E \bigr\}^t$.
Here, $(a,b)$ denotes an interaction in which node $a$ is the initiator
and node $b$ is the responder.
Each element $M(i) = (a_i,b_i)$, for $i \in [t]$ is drawn uniformly at random from the set of all
directed edges obtained by replacing each undirected edge
$\{a,b\} \in E$ with the two arcs $(a,b)$ and $(b,a)$.

\newcommand{\prmt}[1]{\underset{\Vec{m}(t)}{\mathbf{Pr}}\left[#1\right]}
\newcommand{\prmtt}[1]{\underset{\Vec{m}(t-1)}{\mathbf{Pr}}\left[#1\right]}
\newcommand{\Emt}[1]{\underset{\Vec{m}}{\mathbf{E}}\left[#1\right]}
For an easier notation when we condition on the sequence $\Vec{m}(t)$ of all random choices until step $t$, write
\begin{align}
\prmt{\cdot} = \pr{\cdot \mid \Vec{M}(t) = \Vec{m}(t)}.
\end{align}
Recall that \emph{per assumption}, all choices of the scheduler are independent.
Therefore, for each step $t$, directed pair $(a,b) \in E$, and sequence $\Vec{m}(t-1) \in E^{t-1}$, it holds
\begin{align}
\prmtt{M(t) = (a,b)} = \pr{M(t) = (a,b)} = \frac{1}{2|E|}.
\end{align}
For the analysis of our theorem, we furthermore define the following random variables: $B_u(t) \in \{0,1\}$ is the indicator that $u$ is \emph{black} at the beginning of step $t$.
$W_u(t) \in \{0,1\}$ id the indicator that $u$ is \emph{white} at the beginning of step $t$.
$S_u(t) \in \{0,1\}$ is the indicator that $u$ is \emph{black} and interacts with a white node in step $t$.

Having these definitions out of the way, we begin with the analysis.
We divide the analysis into three parts, one for each of the three properties of \Cref{thm:bitsampling}.
In \Cref{sec:uniform}, we show that the sampled bits are uniformly distributed in $\{0,1\}$, i.e., we show the first property.
Then, in \Cref{sec:independent}, we prove the second property and show that the sampled bits are independent.
Finally, we show the third property in \Cref{sec:frequently} and furthermore prove that this property implies that each node samples many bits within $O(|E(G)|/\delta(G) \log n)$ steps.

\subsection{Bits are Uniform}
\label{sec:uniform}

In this section we show the first statement of \cref{thm:bitsampling}.
In fact we show a slightly stronger statement that we will reuse for the other proofs.
We show that, \emph{given the entire history of the protocol up to step $t$}, the following holds: If a random bit is chosen in step $t$, then the probability that the random bit is $0$ or $1$ is uniform.
Formally:
\begin{lemma}
\label{lemma:unbiasedbit}
Fix a node $u$ and a step $t \geq 1$.
Let $\Vec{m}(t-1)$ be a sequence of $t-1$ directed edges, s.t., it holds
$\prmtt{S_u(t)=1 }>0.$
Then,
\begin{align*}
\prmtt{\Rand u(t)=1 \mid S_u(t)=1} =
\prmtt{\Rand u(t)=0 \mid S_u(t)=1} = \tfrac{1}{2}.
\end{align*}
\end{lemma}
\begin{proof}
Recall that $S_u(t)$ indicates that $u$ is black, and we know that at the beginning of step $t$
it has at least one white neighbor.
On a high level, we want to show that conditioning on $u$ being black and any neighbor being white does not affect the probability of drawing a certain bit.
For each $v \in V$, the variables $B_v(t)$ and $W_v(t)$ can be expressed as functions of $\Vec{m}(t-1)$.
Formally, there is a function:
\begin{align*}
f_{B,u}(\vec{m}(t-1)) = \begin{cases}
1 & \text{if } B_u(t)=1,\\
0 & \text{else}.
\end{cases}
\end{align*}
We only prove this for $B_u(t)$, as if a node is not black, it is white.
Hence, if $B_u(t)$ is a function of $M(t-1)$, so is $W_u(t)$.
Node $u$ is black exactly when $u$ was the initiator in its last interaction
before step $t$ (if such an interaction exists).
Let
\begin{align*}
\tau_u(\Vec{m}(t-1)) =
\begin{cases}
\max\{ i < t \mid m(i)=(a,b),\; u\in\{a,b\}\}
& \text{if such $i$ exists},\\
\bot & \text{else}.
\end{cases}
\end{align*}
For readability, denote $\tau_u(\Vec{m}(t-1))$ simply as $\tau$.
With definition , it easy to verify that
\begin{align*}
B_u(t) = f_{B,u}(\Vec{m}(t-1)) =
\begin{cases}
\1\{m(\tau)=(u,a)\}
& \text{if } \tau\neq\bot,\\
0 & \text{else}.
\end{cases}
\end{align*}
Thus, $B_u(t)$ is completely decided by $\Vec{m}(t-1)$ as claimed.

We now define the set of white neighbors of $u$ at the beginning of step $t$ as
\begin{align*}
N_u^{\mathrm{w}}(\Vec{m}(t-1))
=
\{ v\in N_u \mid f_{B,v}(\Vec{m}(t-1))=0 \}.
\end{align*}
Note that, just as $B_v(t)$ and $W_v(t)$, this set is fixed given $\Vec{m}(t-1)$.
Further, note that $|N_u^{\mathrm{w}}(\Vec{m}(t-1))| \geq 1$ as we conditions on an $\Vec{m}(t-1)$ where $\prmtt{S_u(t)=1} \geq 0$ and, therefore, $u$ has at least one white neighbor.

We now describe the events $\{\Rand u(t)=1\}$ and $\{S_u(t)=1\}$ with respect to
the random choice in step $t$.
Conditioned on $\Vec{m}(t-1)$, event $S_u(t)=1$ occurs precisely if
$M(t)=(a_t,b_t)$ with either $a_t=u$ and $b_t\in N_u^{\mathrm{w}}(\Vec{m}(t-1))$,
or $b_t=u$ and $a_t\in N_u^{\mathrm{w}}(\Vec{m}(t-1))$.
Hence,
\begin{align*}
\prmtt{S_u(t)=1 }
&= \sum_{v\in N_u^{\mathrm{w}}(m)}
\left(\prmtt{M(t)=(u,v)} + \prmtt{M(t)=(v,u))}\right).
\end{align*}
Since the edge sampled in step $t$ is independent of $\Vec{m}(t-1)$ and each
orientation of each edge is chosen with probability $1/(2|E|)$, we obtain
\begin{align*}
\prmtt{S_u(t)=1} &= \sum_{v\in N_u^{\mathrm{w}}(m)}
\left(\pr{M(t)=(u,v)} + \pr{M(t)=(v,u))}\right)\\
&= \sum_{v\in N_u^{\mathrm{w}}(m)}
\left(\frac{1}{2|E|}+\frac{1}{2|E|}\right)
= \frac{|N_u^{\mathrm{w}}(\Vec{m}(t-1))|}{|E|}.
\end{align*}
On the other hand, $\Rand u(t)=1$ occurs precisely when
$M(t)=(u,v)$ for some $v\in N_u^{\mathrm{w}}(m)$.
Thus,
\begin{align*}
\prmtt{\Rand u(t)=1 }
&= \sum_{v\in N_u^{\mathrm{w}}(\Vec{m}(t-1))} \prmtt{M(t)=(u,v)}\\
&= \sum_{v\in N_u^{\mathrm{w}}(\Vec{m}(t-1))} \pr{M(t)=(u,v)}
= \frac{|N_u^{\mathrm{w}}(m)|}{2|E|}.
\end{align*}
Since the event $\{\Rand u(t)=1\}$ is included in the event $\{S_u(t)=1\}$, we have $\{\Rand u(t)=1\} \cap \{S_u(t)=1\} = \{\Rand u(t)=1\} $.
By the definition of conditional
probability we conclude
\begin{align*}
\prmtt{\Rand u(t)=1 \mid S_u(t)=1}
&= \frac{\prmtt{\{\Rand u(t)=1\} \cap \{S_u(t)=1\}}}
{\prmtt{S_u(t)=1}}\\
&= \frac{\prmtt{\Rand u(t)=1}}
{\prmtt{S_u(t)=1}}
= \frac{\frac{|N_u^{\mathrm{w}}(\Vec{m}(t-1))|}{2|E|}}
{\frac{N_u^{\mathrm{w}}(\Vec{m}(t-1))}{|E|}}
= \frac{1}{2}. \qedhere
\end{align*}
\end{proof}
By using the law of total probability, this proves the first statement of the theorem and shows that the distribution of each individual bit uniform.

\subsection{Bits are Independent}
\label{sec:independent}

In this section we show the second statement of \cref{thm:bitsampling}.
More precisely, we formally show that all the generated bits are independent of each other.
Intuitively, this follows because the value is determined by orientation of the edge in step $t$ and this orientation is used only for \emph{one} random bit (with other node receiving $\bot$ in the step).
Given that the orientations of edges in different steps are independently chosen by the scheduler, the independence of our bit follows.
To make it more formal, we will show that in any possible order which the bits are picked, they are always independent of past and
future bits.
Intuitively, we will show that, given that in step $t$ where are all previous random bits match a fixed bitstrings, the probability that the next bit matches is $\frac{1}{2}$.
This is true because the distribution of each new bit is uniform independent of the entire processes' history including all previously sampled bits.
This follows directly from \cref{lemma:unbiasedbit}.

Given this intuition, we now want show it formally in the following lemma
\begin{lemma}
\label{lemma:independentbits}
For each $u \in V$ and $\ell \geq 1$, let $R_u(\ell) \in \{0,1\}$ be the
$\ell$-th bit sampled by $u$. Then all random variables
$(R_u(\ell))_{u \in V, \ell \geq 1}$ are independent and uniformly distributed.
\end{lemma}

\begin{proof}
We prove that every finite collection of sampled bits is independent and
uniformly distributed. For each agent $u \in V$, let $s_u \geq 0$ be an
integer, and fix an arbitrary bitstring
$b_u(1),\ldots,b_u(s_u) \in \{0,1\}$ of length $s_u$.
For each agent $u$ and step $t$, let $S_u(t)\in\{0,1\}$ indicate whether
$u$ samples a bit in step $t$, and let
$\mathcal{S}_u(t)=\sum_{\tau=1}^{t} S_u(\tau)$ be the number of bits sampled
by $u$ up to step $t$.
For each $u \in V$ and step $t$, define the indicator
\begin{align*}
Y_u(t)
=
\begin{cases}
1
& \text{if $\mathcal{S}_u(t)>s_u$,}\\
1
& \text{if $S_u(t)=0$,}\\
1
& \text{if $S_u(t)=1$ and
$R_u(\mathcal{S}_u(t))=b_u(\mathcal{S}_u(t))$,}\\
0
& \text{otherwise.}
\end{cases}
\end{align*}
Thus, $Y_u(t)=0$ only if $u$ samples one of the first $s_u$ bits in step $t$
and this bit does not match the prescribed bitstring. Once $u$ has already
sampled more than $s_u$ bits, we always set $Y_u(t)=1$.
Furthermore, define
$
Z(t)
=
\prod_{u \in V}
\prod_{\tau=1}^{t} Y_u(\tau)$.
Then $Z(t)=1$ if and only if all relevant bits sampled up to step $t$ match
their prescribed bitstrings, that is, $Z(t)=1$ if and only if for all
$u \in V$ and all $j \leq \min\{\mathcal{S}_u(t),s_u\}$, it holds that
$R_u(j)=b_u(j)$.

Let $T$ be the first step at which every agent $u$ has sampled at least
$s_u$ bits. It is enough to show that $\pr{Z(T)=1}=2^{-\sum_{u \in V}s_u}$.
Indeed, since the integers $s_u$ and the target bits $b_u(j)$ were chosen
arbitrarily, this implies that every finite collection of sampled bits is
independent and uniformly distributed.

To this end, for every $t \geq 0$ define $Z'(t) = 2^{\sum_{u \in V}\min\{\mathcal{S}_u(t),s_u\}} \cdot Z(t)$
Given this definition, we prove the stronger statement that, for every $t \geq 0$,
\begin{align}
\E{Z'(t)}= 1 \label{ineq:martingale}.
\end{align}
Applying this statement at time $T$, we have
$\min\{\mathcal{S}_u(T),s_u\}=s_u$ for every $u \in V$. Therefore,
\begin{align*}
1=\E{Z'(T)}=\E{2^{\sum_{u \in V}s_u}Z(T)}=2^{\sum_{u \in V}s_u} \cdot
\E{Z(T)} = 2^{\sum_{u \in V}s_u} \cdot
\pr{Z(T)=1}.
\end{align*}
Note that the last equality follows since $Z(T) \in \{0,1\}$ is an indicator. Thus, solving the previous for $\pr{Z(T)=1}$ gives
\begin{align*}
\pr{Z(T)=1}
=
2^{-\sum_{u \in V}s_u}.
\end{align*}
This proves that every finite collection of sampled bits is independent and
uniformly distributed. 

It remains to prove \cref{ineq:martingale}.
We will the statement by induction on $t$.
For $t=0$, this is immediate, since $Z(0)=1$ and
$\mathcal{S}_u(0)=0$ for all $u \in V$.
Now assume that the statement holds for $t-1$. Let
$V^{\star}(t)=\{u \in V \colon \mathcal{S}_u(t-1)<s_u\}$ be the set of agents
that have not yet sampled all their prescribed bits before step $t$.
Then, it easy to verify that it holds
\begin{align*}
Z'(t)
=
Z'(t-1)
\cdot
\prod_{u \in V^{\star}(t)} 2^{S_u(t)}Y_u(t).
\end{align*}
In the following, we condition on the full configuration $\Vec{M}(t-1)=\vec{m}$ before step
$t$. By the law of total expectation,
\begin{align*}
\E{
Z(t)
}
&=
\sum_{\vec{m}}
\pr{\Vec{M}(t-1)=\vec{m}}
\cdot
\Emt{
Z(t-1)
\cdot
\prod_{u \in V^{\star}(t)} 2^{S_u(t)}Y_u(t)
}\\
&=
\sum_{\vec{m}}
\pr{\Vec{M}(t-1)=\vec{m}}
\cdot
\Emt{
Z(t-1)
}
\cdot
\Emt{
\prod_{u \in V^{\star}(t)} 2^{S_u(t)}Y_u(t)
}.
\end{align*}
The second equality follows as $Z'(t-1)$ is completely determined by $\Vec{m}$ and therefore $\Emt{Z(t-1)}$ is a constant and not a random variable.
It remains to evaluate the last conditional expectation. By assumption, in
each step at most one agent samples a bit, so
$\sum_{u \in V^{\star}(t)} S_u(t) \in \{0,1\}$. If no agent in
$V^{\star}(t)$ samples a bit in step $t$, then $S_u(t)=0$ and $Y_u(t)=1$
for all $u \in V^{\star}(t)$, and therefore
$\prod_{u \in V^{\star}(t)} 2^{S_u(t)}Y_u(t)=1$.
Otherwise, there is a unique agent $u^\star \in V^{\star}(t)$ with
$S_{u^\star}(t)=1$. Conditioned on $\Vec{M}(t-1)=\vec{m}$ and on the event
that $u^\star$ samples a bit in step $t$, the sampled bit is uniform and
independent of all previously sampled bits. Thus, by \cref{lemma:unbiasedbit}, it holds
\begin{align*}
\prmt{
R_{u^\star}(\mathcal{S}_{u^\star}(t))
=
b_{u^\star}(\mathcal{S}_{u^\star}(t))
\mid
S_{u^\star}(t) = 1
}
=
\frac{1}{2}.
\end{align*}
Consequently,
\begin{align*}
\Emt{
2^{S_{u^\star}(t)}Y_{u^\star}(t)
\mid
S_{u^\star}(t) = 1
}
=
2 \cdot \frac{1}{2}
=
1.
\end{align*}
All other factors in the product are equal to $1$. Hence, in both cases,
\[\Emt{\prod_{u \in V^{\star}(t)} 2^{S_u(t)}Y_u(t)}=1.\]
Plugging this back into the previous display gives
\begin{align*}
\E{
Z(t)
}
&=
\sum_{\vec{m}}
\pr{\Vec{M}(t-1)=\vec{m}}
\cdot
\Emt{
Z(t-1)
}
\cdot
\Emt{
\prod_{u \in V^{\star}(t)} 2^{S_u(t)}Y_u(t)
}\\
&=
\sum_{\vec{m}}
\pr{\Vec{M}(t-1)=\vec{m}}
\cdot
\Emt{
Z(t-1)
} = \E{Z'(t-1)} = 1.
\end{align*}
where the last equality follows from the induction hypothesis.
\end{proof}

\subsection{Bits are Sampled Frequently}
\label{sec:frequently}
In this section, we prove the third statement of \cref{thm:bitsampling} and show that each node will regularly generate new random bits.
Recall that this property is crucial as, up until now, we only argued that all bits a node gets are unbiased but never argued that it will actually get bits.
To conclude our analysis of \cref{thm:bitsampling}, we therefore show the following lemma.
\begin{lemma}
\label{lem:timewhatistime}
For each node, let $\mathcal{T}_u$ be expected time in which the node interacts once on expectation, i.e., $\mathcal{T}_u = \frac{|E|}{|N_u|}$. Further, let $\mathcal{T} = \max_{u \in V} \mathcal{T}_u = \tfrac{|E(G)|}{\delta(G)}$.
Fix a step $t$. Then, for any $m \in E^t$, it holds:
\begin{align*}
\pr{\sum_{t'=t}^{t+4\mathcal{T}} S_v(t') > 0
\mid \Vec{M}(t)=\Vec{m}(t)}
\geq \frac{1}{96}.
\end{align*}
\end{lemma}

\begin{proof}
For easier notation, we assume w.l.o.g. that the process starts in step $0$
and omit the conditioning on $\{\Vec{M}(t)=\Vec{m}(t)\}$.
For each $x \in N_v$ and each step $t \geq 0$, let $W_x(t)$ and $B_v(t)$
be the indicators that $x$ is white and $v$ is black in step $t$, respectively.
Further, let $A_x(t) \in \{0,1\}$ be the indicator that $x$ interacts in step
$t$. Then
\begin{align}
\E{\sum_{t=0}^{4\mathcal{T}} S_v(t)}
&=
\E{\sum_{t=0}^{4\mathcal{T}}
\sum_{x \in N_v}
(A_v(t) \cdot A_x(t)) \cdot (B_v(t) \cdot W_x(t))} \label{ineq:def} \\
&=
\sum_{t=0}^{4\mathcal{T}}
\sum_{x \in N_v}
\E{A_v(t) \cdot A_x(t)}
\cdot
\E{B_v(t) \cdot W_x(t)} \label{ineq:indep} \\
&=
\frac{1}{|E|}
\sum_{t=0}^{4\mathcal{T}}
\sum_{x \in N_v}
\E{B_v(t) \cdot W_x(t)} . \label{ineq:uni}
\end{align}
Here, \cref{ineq:def} follows from the definition of $S_v(t)$,
\cref{ineq:indep} follows because the scheduler chooses the interacting pair
independently of the colors, and \cref{ineq:uni} follows since every edge is
picked uniformly at random.
By \cref{lemma:bw} (that we will show below), for all $t \geq 2\mathcal{T}$ and all $x \in N_v$,
we have $\pr{B_v(t) \cdot W_x(t) = 1} \geq \frac{1}{8}$.
Hence, using only the interval $\{2\mathcal{T},\ldots,4\mathcal{T}-1\}$,
we obtain
\begin{align}
\E{\sum_{t=0}^{4\mathcal{T}} S_v(t)}
&\geq
\frac{1}{|E|}
\sum_{t=2\mathcal{T}}^{4\mathcal{T}-1}
\sum_{x \in N_v}
\pr{B_v(t) \cdot W_x(t)=1}
\geq
\frac{1}{|E|} \cdot 2\mathcal{T} \cdot |N_v| \cdot \frac{1}{8}
=
\frac{|N_v|}{4\delta},
\label{ineq:firstmoment}
\end{align}
where we used $\mathcal{T}=|E|/\delta$.
It remains to lower-bound the probability that at least one such successful
interaction occurs. Define
\begin{align*}
X =
\sum_{t=2\mathcal{T}}^{4\mathcal{T}-1} S_v(t) && \text{and} && Y =
\sum_{t=2\mathcal{T}}^{4\mathcal{T}-1} A_v(t)
\end{align*}
as the number of steps in this interval in which $v$ samples a bit and interacts respectively. Since $v$ independently
interacts in any fixed step with probability $|N_v|/|E|$, this means that $Y$ is binomially distributed.
Therefore,
\begin{align*}
\E{Y} &= 2\mathcal{T}\frac{|N_v|}{|E|}=2\frac{|N_v|}{\delta}
\intertext{and}
\mathbf{Var}\left[{Y}\right] &= 2\mathcal{T}\left(\frac{|N_v|}{|E|}\right)\left(1-\frac{|N_v|}{|E|}\right)=2\frac{|N_v|}{\delta}\left(1-\frac{|N_v|}{|E|}\right) < \E{Y}.
\end{align*}
Using the definition of variance, we get
\begin{align*}
\E{Y^2}
=
\mathbf{Var}[Y] + \E{Y}^2
\leq
\E{Y}+\E{Y}^2
=
2\frac{|N_v|}{\delta}
+
4\frac{|N_v|^2}{\delta^2}.
\end{align*}
Recall $X \leq Y$ because every successful sampling requires an interaction. Furthermore, both $X$ and $Y$ are non-negative. This implies $\E{X^2} \leq \E{Y^2}$ and therefore
\begin{align}
\E{X^2}
\leq
\E{Y^2}
\leq
2\frac{|N_v|}{\delta}
+
4\frac{|N_v|^2}{\delta^2}.
\label{ineq:secondmoment}
\end{align}
Finally, note that $X \geq 0$ and $X$ has finite variance.
Then, by the Paley-Zygmund inequality (see \cref{paley}) , applied with threshold $0$,
\begin{align*}
\pr{\sum_{t=0}^{4\mathcal{T}} S_v(t) > 0} &\geq \pr{\sum_{t=2\mathcal{T}}^{4\mathcal{T}} S_v(t) > 0} =
\pr{X>0}\\
&\geq
\frac{\E{X}^2}{\E{X^2}}
\geq
\frac{
\left(\frac{|N_v|}{4\delta}\right)^2
}{
2\frac{|N_v|}{\delta}
+
4\frac{|N_v|^2}{\delta^2}
}
=
\frac{
\frac{|N_v|}{\delta}
}{
32 + 64\frac{|N_v|}{\delta}
} \geq \frac{1}{96}.
\end{align*}
This proves the lemma, assuming \cref{lemma:bw}.
\end{proof}

\begin{lemma}
\label{lemma:bw}
For all $t \geq 2\mathcal{T}$ and all $x \in N_v$, it holds that
\begin{align*}
\pr{W_x(t) \cdot B_v(t)=1} \geq \frac{1}{8}.
\end{align*}
\end{lemma}

\begin{proof}
For each agent $y \in V$ and each step $t>0$, define $t_y^{\star}$ as
the last step before $t$ in which $y$ was active, or $\infty$ if $y$
was inactive in all steps before $t$:
\begin{align*}
t_y^{\star}
=
\begin{cases}
\max\{t' \in \{0,\ldots,t-1\}: A_y(t')=1\}
& \text{if such a step exists},\\
\infty
& \text{otherwise}.
\end{cases}
\end{align*}
Let $\mathfrak{T}
=
\{t_v^{\star}\neq \infty,\ t_x^{\star}\neq \infty\}$
be the event that both $v$ and $x$ have interacted at least once before step
$t$.
First, we lower-bound the probability of $\mathfrak{T}$. For a fixed node
$u$, the probability that it does not interact during the first
$2\mathcal{T}$ steps is
\begin{align*}
\left(1-\frac{|N_u|}{|E|}\right)^{2\mathcal{T}}
\leq
\left(1-\frac{\delta}{|E|}\right)^{2\mathcal{T}}
=
\left(1-\frac{1}{\mathcal{T}}\right)^{2\mathcal{T}}
\leq
e^{-2}.
\end{align*}
Thus, by the union bound, for every $t \geq 2\mathcal{T}$,
\begin{align}
\pr{\mathfrak{T}}
&\geq
1
-
\pr{t_v^{\star}=\infty}
-
\pr{t_x^{\star}=\infty}\geq
1-2e^{-2}
\geq
\frac{1}{2}.
\label{ineq:tbound}
\end{align}

Conditioned on $\mathfrak{T}$, the last interaction of $v$ before step
$t$ determines its current color. In this last interaction, $v$ is the
initiator with probability $1/2$. Hence
\begin{align}
\pr{B_v(t)=1 \mid \mathfrak{T}}
=
\frac{1}{2}.
\label{ineq:black}
\end{align}
Here, we used $\mathfrak{T}$ and therefore $t^\star_x$ and $t^\star_v$ are only determined by the interaction pairs until step $t$ but not the initiators and responders.

It remains to lower-bound the probability that $x$ is white, conditioned on
$v$ being black and on $\mathfrak{T}$. If $t_x^{\star}=t_v^{\star}$, then
the last interaction of both nodes was their mutual interaction. Since $v$ is
black in that interaction, $x$ is white with probability $1$. Otherwise,
if $t_x^{\star}\neq t_v^{\star}$, then the last interaction of $x$ is
different from the last interaction of $v$, and $x$ is white with
probability $1/2$. Therefore
\begin{align}
\pr{W_x(t)=1 \mid B_v(t)=1,\mathfrak{T}}
&=
\pr{t_x^{\star}=t_v^{\star} \mid B_v(t)=1,\mathfrak{T}}
\cdot 1 \nonumber \\
&\phantom{=}
+
\pr{t_x^{\star}\neq t_v^{\star} \mid B_v(t)=1,\mathfrak{T}}
\cdot \frac{1}{2}
\geq
\frac{1}{2}.
\label{ineq:white}
\end{align}

Combining \cref{ineq:tbound}, \cref{ineq:black}, and \cref{ineq:white}, we get
\begin{align*}
\pr{W_x(t)\cdot B_v(t)=1}
&\geq
\pr{\mathfrak{T}}
\cdot
\pr{B_v(t)=1 \mid \mathfrak{T}}
\cdot
\pr{W_x(t)=1 \mid B_v(t)=1,\mathfrak{T}} \\
&\geq
\frac{1}{2}\cdot \frac{1}{2}\cdot \frac{1}{2}
=
\frac{1}{8}.
\end{align*}
This proves the lemma.
\end{proof}
Finally, we show that the third property of \Cref{thm:bitsampling} implies that, w.h.p., all nodes will nodes will repeatedly sample unbiased and independent random bits.
Formally, we show the following corollary.

\restatebitsamplingRepeat*
\begin{proof}
First, we note that the second statement follows directly from the first.
Assume that the first statement holds. Choosing $\ell = (c+1)\log n$ gives
each node $u$ a probability of at least $1-e^{-\ell}=1-n^{-(c+1)}$ to sample
$\ell$ bits. Thus, with probability at most $n^{-(c+1)}$, node $u$ samples
less than $\ell$ bits. Denote this event by $F_u$.
Note that the events $(F_u)_{u \in V}$ are not independent. However, by the
union bound, the probability that any of these $n$ events happens is at most $\frac{1}{n^c}$.
Therefore, it is sufficient to prove the first statement.

For the proof, fix an agent $u$.
We divide the time into non-overlapping phases
of length $4{|E(G)|}/{\delta(G)}$. More precisely, for each $i \geq 1$, phase
$i$ consists of the steps $\left[
(i-1)\cdot 4\frac{|E(G)|}{\delta(G)},
i\cdot 4\frac{|E(G)|}{\delta(G)}-1
\right]$.
We show that within $L = 96 C \ell$
phases, where $C \geq 4$ is a sufficiently large constant, agent $u$ samples
at least $\ell$ bits with probability at least $1-e^{-\ell}$.

For each phase $i \in \{1,\ldots,L\}$, let $X_i \in \{0,1\}$ be the random
variable that indicates whether $u$ successfully samples a bit in phase $i$.
If $X_i=1$, we call phase $i$ successful. Since each successful phase produces
at least one random bit, the number of successful phases is a lower bound for
the number of sampled bits.
It remains to show that $\mathcal{X}= \sum_{i=1}^{L} X_i \geq \ell$ with probability at least
$1-e^{-\ell}$. As each phase has length $4{|E(G)|}/{\delta(G)}$, the total number of steps is in $O(\ell \frac{|E(G)|}{\delta(G)})$, which proves the statement.

By property $(3)$ of
\Cref{thm:bitsampling}, for every possible configuration
$\Vec{m}$ at the beginning of phase $i$ the probability of success is at least $\frac{1}{96}$.
Thus, we have for every $i \geq 1$ and every vector $(x_{i-1}, \ldots, x_1) \in \{0, 1\}^{i-1}$ that
\begin{align*}
\E{X_i \mid X_{i-1} = x_{i-1}, \ldots, X_{1} = x_{1} }
\geq
\frac{1}{96}.
\end{align*}
Moreover, this implies
\begin{align*}
\E{\prod_{i=1}^L X_i } \geq
(\frac{1}{96})^L &&\text{and}&& \E{\mathcal{X}}=\sum_{i=1}^{L} \E{X_i}\geq\frac{L}{96}=C\ell.
\end{align*}
Although the random variables $X_1,\ldots,X_L$ are not independent,
the conditional lower bound above is sufficient to apply the multiplicative
Chernoff bound with $\mu = C\ell$ as our variables dominate Bernoulli distributed values with parameter $p = \frac{1}{96}$. Hence, for
every $\alpha \in (0,1)$,
\begin{align*}
\pr{\mathcal{X}\leq (1-\alpha)C\ell}
\leq
\exp\left(-\frac{\alpha^2 C \ell}{2}\right).
\end{align*}
If we choose $\alpha = 1-\frac{1}{C}$, we have
$ (1-\alpha)\frac{L}{96}
=
\frac{1}{C}\cdot C\ell
=
\ell.$
Thus,
\begin{align*}
\pr{\mathcal{X}\leq \ell}
&\leq
\exp\left(
-\frac{\left(1-\frac{1}{C}\right)^2 C\ell}{2}
\right).
\end{align*}
For every $C \geq 4$, we have
$
\frac{\left(1-\frac{1}{C}\right)^2 C}{2}
\geq
1$,
it holds,
$
\pr{\mathcal{X}\leq \ell}
\leq
e^{-\ell}.
$
as desired.
Since the number of successful phases is a lower bound for the number of bits
sampled by $u$, it follows that $u$ samples at least $\ell$ bits within
$L=96C\ell$ phases with probability at least $1-e^{-\ell}$.
This proves the first statement
and the corollary.
\end{proof}

\section{Baseline: Space-Efficient Uniform Counting}
\label{sec:folklore}

In this appendix we present \Baseline, a simple uniform population protocol that counts the population within $O(\hittingtime(G) \cdot n \log n)$ interactions. Recall that $\hittingtime(G)$ denotes the maximal hitting time of a random walk on $G$. The protocol is a direct extension of the leader-election protocol of Beauquier, Blanchard, and Burman \cite{DBLP:conf/opodis/BeauquierBB13}. For its analysis, we use results of Alistarh, Rybicki, and Voitovych \cite{DBLP:journals/dc/AlistarhRV25}.

At a high level, \Baseline{} works as follows. Each agent is either \emph{active} or \emph{inactive}, as indicated by the field $\isactive{}\in\{0,1\}$. An active agent holds one or more \emph{tokens} and stores its current token count. An inactive agent holds no tokens and stores a the highest number of tokens it observed so far. 
We represent both quantities by a positive integer field $\val{}$. The meaning of this field depends on whether the agent is active or inactive. For an active agent, $\val{}$ is its current token count. For an inactive agent, $\val{}$ is its current estimate, which is propagated using the maximum rule described below. In either case, the output of an agent $u$ is $\val{u}$. Thus, the protocol stabilizes when $\{\val{u}=n\}$ for every $u\in V$.

Initially, every agent is active and holds exactly one token. The transition rules are as follows. 
\begin{itemize}
\item Suppose that two active agents $u$ and $v$, holding $\val{u}$ and $\val{v}$ tokens, respectively, interact, and assume without loss of generality that $u$ is the initiator. Agent $v$ transfers all its tokens to $u$ and becomes inactive. After the interaction, $u$ is active with value $\val{u}+\val{v}$, while $v$ is inactive with the same value $\val{u}+\val{v}$.

\item If two inactive agents $u$ and $v$, storing values $\val{u}$ and $\val{v}$, respectively, interact, both update their stored value to $\max\{\val{u},\val{v}\}$.

\item Suppose that an active agent $u$, holding $\val{u}$ tokens, interacts with an inactive agent $v$, storing value $\val{v}$. The agents exchange roles: agent $v$ becomes active and receives the $\val{u}$ tokens held by $u$, while $u$ becomes inactive. After the interaction, $v$ stores the previous value $\val{u}$, while $u$ stores $\max\{\val{u},\val{v}\}$.
\end{itemize}

\begin{lstlisting}[caption={\Baseline($u$, $v$)}, output={$\rho(\isactive{u},\val{u})=\val{u}$}, label={alg:baseline}]
$m \gets \max\{\val v, \val u\}$

if $\isactive u = \isactive v = 1$ then /* If both are active, move all tokens to the initiator */
    $(\isactive u, \isactive v) \gets (1,0)$
    $(\val u, \val v) \gets (\val v + \val u, \val v + \val u)$
    
else if $\isactive u = \isactive v = 0$ then /* If both are inactive, propagate the largest value */
    $(\val u, \val v) \gets (m,m)$
    
else if $\isactive u =1$  /* Otherwise, swap states and propagate largest value to inactive agent */
    $(\isactive u, \isactive v) \gets (\isactive v,\isactive u)$
    $(\val u, \val v) \gets (m,\val u)$

else 
    $(\isactive u, \isactive v) \gets (\isactive v,\isactive u)$
    $(\val u, \val v) \gets (\val v, m)$
\end{lstlisting}

\subsection{Analysis of \Baseline}

We now analyze \Baseline{} and prove the following theorem.

\begin{theorem}
\label{thm:baseline-counting}
\Baseline{} is a uniform population protocol using $O(n)$ states on populations of size $n$ such that, on every connected interaction graph $G=(V,E)$ with $|V|=n$, all agents eventually stabilize to output $n$. Moreover, with high probability, stabilization occurs within $O(\hittingtime(G) \cdot n \log n)$ interactions.
\end{theorem}

\begin{proof}
We first establish correctness. The protocol maintains the invariant that the token counts of all active agents sum to $n$. This invariant holds initially because every agent holds one token. An interaction between two active agents transfers tokens between them but neither creates nor destroys any tokens. An interaction between an active and an inactive agent only moves an existing collection of tokens from one agent to another. Finally, an interaction between two inactive agents does not affect any tokens. Hence, every transition preserves the invariant.

An interaction between two active agents decreases the number of active agents by one. Every other interaction leaves the number of active agents unchanged. Thus, the number of active agents never increases. It can never become zero, because the active agents collectively hold all $n$ tokens. Consequently, once only one active agent $s$ remains, the invariant implies that $s$ holds all $n$ tokens and therefore stores the value $n$.

From that point onward, the value $n$ spreads as an epidemic. Call an agent \emph{informed} if its stored value is $n$. If two inactive agents interact and one of them is informed, the maximum rule makes both agents informed. If the active agent interacts with an inactive agent, then the active agent transfers all $n$ tokens to the inactive agent and stores $n$ upon becoming inactive. Thus, both agents are informed after the interaction. In particular, once only one active agent remains, an informed agent can never become uninformed. Because $G$ is connected and every edge is selected infinitely often with probability one, every agent is eventually informed. Hence, all agents eventually stabilize to output $n$.

We next bound the time complexity. We use Claim \ref{clm:single-active}, which is proved immediately after this theorem. By Claim \ref{clm:single-active}, with high probability a unique active agent $s$ remains within $O(\hittingtime(G) \cdot n \log n)$ interactions. As argued above, this agent stores the value $n$. The epidemic then informs every agent within an additional $O(\broadcasttime(G))$ interactions, with high probability. Once all agents store $n$, every subsequent transition preserves their outputs, so the protocol has stabilized. Using the standard inequality $\broadcasttime(G)\leq\hittingtime(G)$, the dissemination time is dominated by $O(\hittingtime(G) \cdot n \log n)$.

Finally, we establish the state complexity. Each agent stores the bit $\isactive{}$, indicating whether it is active or inactive, and a positive integer $\val{}$. No value reachable in an execution on $n$ agents can exceed $n$. Indeed, by the token invariant, an active agent never holds more than $n$ tokens. Moreover, every value stored by an inactive agent is equal to a token count held by an active agent at the same or an earlier interaction. This property follows by induction over the interactions. An active--active interaction assigns the newly inactive agent the new token count of the surviving active agent, while the other two rules either copy an active token count or retain the maximum of two previously stored estimates. Therefore, the value of an inactive agent cannot exceed $n$ either.

It follows that, on a population of size $n$, every reachable state consists of one bit and one value in $[1,n]$. Thus, at most $2n$ states are reachable. The transition rules themselves do not depend on $n$, so the protocol is uniform. Together with Claim \ref{clm:single-active}, this completes the proof.
\end{proof}

We now prove the claim used in the time analysis.

\begin{claim}
\label{clm:single-active}
With high probability, exactly one active agent remains after $O(\hittingtime(G) \cdot n \log n)$ interactions.
\end{claim}

\begin{proof}
Consider the leader-election protocol of Beauquier, Blanchard, and Burman \cite{DBLP:conf/opodis/BeauquierBB13}, together with the analysis of Alistarh, Rybicki, and Voitovych \cite{DBLP:journals/dc/AlistarhRV25}. In that protocol, every agent is initially black. When two black agents interact, the initiator remains black and the responder becomes white. When a black agent and a white agent interact, they exchange colors. An interaction between two white agents does not change their colors.

We couple the two protocols by using the same sequence of ordered interactions and identifying active agents with black agents and inactive agents with white agents. Initially, this correspondence holds because every agent is both active in our protocol and black in the leader-election protocol. An active--active interaction has the same effect on activity statuses as a black--black interaction has on colors. Similarly, an active--inactive interaction swaps the two activity statuses, just as a black--white interaction swaps the two colors. Finally, neither an inactive--inactive interaction nor a white--white interaction changes the corresponding statuses. The correspondence therefore holds after every interaction.

It follows that the number of active agents in our protocol is always equal to the number of black agents in the leader-election protocol. The analysis in \cite{DBLP:journals/dc/AlistarhRV25} shows that, with high probability, only one black agent remains after $O(\hittingtime(G) \cdot n \log n)$ interactions. Therefore, with high probability, exactly one active agent remains after the same number of interactions.
\end{proof}

\section{Tools from Probability Theory}
\begin{theorem}[Chernoff bound ]
Let $X_1,\dots,X_n$ be independent Bernoulli random variables, and define
$X = \sum_{i=1}^n X_i$ and
$\mu = \E{X}$.
For every $0 < \delta < 1$,
\begin{align*}
\pr{X \le (1-\delta)\mu}
\le
\left(\frac{e^{-\delta}}{(1-\delta)^{1-\delta}}\right)^\mu \le
\exp\!\left(-\frac{\delta^2\mu}{2}\right).
\end{align*}
\end{theorem}
\begin{theorem}[Paley--Zygmund inequality]
\label{paley}
Let \(Z \ge 0\) be a random variable with finite variance.
Then, for every $\theta \in [0,1]$,
\[
\pr{Z > \theta \cdot \mathbb{E}[Z]}
\ge
(1-\theta)^2
\cdot \frac{\E{Z}^2}{\E{Z^2}} .
\]
\end{theorem}

\printbibliography

\end{document}